\documentclass[preprint,12pt]{elsarticle}
\usepackage{amsthm}
\theoremstyle{plain}
\newtheorem{theorem}{Theorem}
\newtheorem{lemma}{Lemma}
\newtheorem{proposition}{Proposition}
\newtheorem{corollary}{Corollary}
\theoremstyle{definition}
\newtheorem{definition}{Definition}
\newtheorem{example}{Example}

\usepackage[T1]{fontenc}
\usepackage[utf8]{inputenc}
\usepackage{amsmath,amssymb}
\usepackage{microtype}
\usepackage[hidelinks]{hyperref}
\usepackage{url}
\usepackage{array,tabularx}

\newcommand{\rot}{\mathrel{\sim_{\mathrm{rot}}}}
\newcommand{\cyc}[1]{[\![#1]\!]}
\newcommand{\cstep}{\mathrel{\Rightarrow}}
\newcommand{\eps}{\varepsilon}
\newcommand{\enc}{\varphi}
\newcommand{\B}{\{0,1\}}
\newcommand{\N}{\mathbb N}
\newcommand{\Eset}{\mathcal E}
\newcommand{\Iset}{\mathcal I}
\newcommand{\Dbar}{\overline D}
\newcommand{\Ctrl}{\mathcal K}
\allowdisplaybreaks

\journal{Theoretical Computer Science}

\begin{document}

\begin{frontmatter}

\title{Undecidability of Confluence for Binary Length-Reducing Cycle Rewriting}
\author{Graham Campbell}
\ead{hello@gjcampbell.co.uk}
\affiliation{country={United Kingdom}}

\begin{abstract}
Confluence guarantees that diverging rewrite choices can always be rejoined. For finite terminating string- and term-rewriting systems, confluence is decidable by critical-pair analysis, and in polynomial time for length-reducing strings. For finite terminating (hyper)graph transformation systems, in contrast, confluence is undecidable. We show that undecidability already appears for words on a circle, that is, strings up to rotation. Confluence of finite cycle-rewriting systems over the fixed alphabet $\{0,1\}$ is undecidable, indeed $\Pi^0_1$-complete, even when every rule has a nonempty right-hand side and strictly reduces length. Under this restriction termination is syntactically evident, and derivations from a nonempty length-$n$ cycle have fewer than $n$ steps. The same holds over every fixed alphabet with at least two letters, while the one-letter case is decidable. Rotation alone separates cyclic from string rewriting. The proof compiles a deterministic verifier into a weighted cycle system with one controlled branch, then into a binary length-reducing system via a run-length code whose cleanup rules send every reducible malformed cycle to one error normal form.
\end{abstract}

\begin{keyword}
cycle rewriting \sep confluence \sep undecidability \sep string rewriting \sep arithmetical hierarchy \sep length-reducing rewriting
\end{keyword}

\end{frontmatter}

\section{Introduction}\label{sec:introduction}

Confluence is the guarantee that diverging rewrite choices can always be rejoined. Together with termination it yields unique normal forms. Classically it is under finite local control. A finite terminating term-rewriting system is confluent exactly when its critical pairs are joinable \cite{KnuthBendix1970,Huet1980}, the same analysis applies to string rewriting \cite{BookOtto1993}, confluence of ground term rewriting is decidable even without termination \cite{DauchetHeuillardLescanneTison1990}, and for finite length-reducing string systems it is decidable in polynomial time \cite{BookODunlaing1981}.

For graph transformation the picture reverses. Critical pairs yield only sufficient conditions, and Plump proved confluence undecidable for finite terminating hypergraph rewriting and graph transformation systems \cite{Plump1993,Plump2005}. Somewhere between words and graphs, finite local analysis therefore stops determining global confluence. This paper shows that the reversal already occurs for cyclic words.

A cyclic word is a string up to rotation, a word on a circle with no distinguished start and no branching. Cycle rewriting applies finite string rules along arcs of the circle, arguably the simplest rewriting formalism beyond strings. We show that rotation alone destroys decidability. Confluence of finite cycle-rewriting systems over the fixed alphabet $\B$ is $\Pi^0_1$-complete (Theorem~\ref{thm:pi-complete}), even when every right-hand side is nonempty and every rule strictly reduces length. Under this promise termination is immediate, derivations from a nonempty length-$n$ cycle have fewer than $n$ steps, and reachability and joinability from a fixed starting cycle are decidable. The undecidability stems neither from termination nor from long computations, but from quantification over all starting cycles, finitely many overlaps embedded in unboundedly many cyclic contexts. Here the critical-pair method fails. On a circle, joinability of an overlap is not stable under enlarging the context, because a later rewrite may travel around the cycle and re-enter it (Example~\ref{ex:context}).

A finite strictly length-reducing rewrite system can serve as an executable specification, with immediate termination and linearly bounded derivations. Our theorem shows that, under cyclic interpretation, there is nevertheless no complete procedure for deciding whether such a specification defines unique normal forms. The same syntactic promise that admits a polynomial-time confluence test on strings admits no confluence algorithm on cycles. In particular, the ordinary finite critical-pair test is incomplete here. Practical procedures must therefore restrict their input class or settle for incomplete criteria.

The result holds over every fixed alphabet with at least two letters (Corollary~\ref{cor:all-alphabets}), while the one-letter case is decidable (Theorem~\ref{thm:unary}), so two letters are optimal. A bounded-shape variant keeps $\Pi^0_1$-completeness for certified nonerasing weight-decreasing systems whose left-hand sides have length at most five, whose right-hand sides have length at most four, and whose symbol weights lie in $\{1,2,3,7,12\}$, the alphabet and rule count now growing instead (Proposition~\ref{prop:bounded-shape}).

The proof chains three equivalences, each of independent interest. Writing $W_e$ for the $e$th computably enumerable set, $M_{e,x}$ for the compiled verifier, $(S_{e,x},\omega_{e,x})$ for the weighted cycle system, and $\mathsf B$ for the binary compiler:
\[
  x\notin W_e
  \overset{\text{Thm.~\ref{thm:verifier-source}}}{\Longleftrightarrow}
  L(M_{e,x})=\varnothing
  \overset{\text{Thm.~\ref{thm:weighted-fork}}}{\Longleftrightarrow}
  S_{e,x}\text{ confluent}
  \overset{\text{Thm.~\ref{thm:binary-transfer}}}{\Longleftrightarrow}
  \mathsf B(S_{e,x},\omega_{e,x})\text{ confluent.}
\]
Section~\ref{sec:verifier} constructs a verifier normal form (Theorem~\ref{thm:verifier-source}) whose nonemptiness problem is undecidable although the machines are deterministic and local, with fixed radius two, a fixed five-symbol alphabet, and configurations that shrink at every step. Theorem~\ref{thm:weighted-fork} (Section~\ref{sec:source}) wraps such a machine into a weighted cycle system whose confluence, over every cycle rather than over intended inputs alone, is equivalent to emptiness of the machine. Theorem~\ref{thm:binary-transfer} (Section~\ref{sec:compiler}) compiles any finite nonerasing weight-decreasing cycle system over a nonempty alphabet into a binary strictly length-reducing one, preserving and reflecting confluence. Section~\ref{sec:classification} combines the three with absorber rules and a unary commutation bound to complete the classification, and Section~\ref{sec:overview} states the invariant governing each stage, the obstructions to a more direct reduction, and the single point at which nonjoinable branching can arise.

Cycle rewriting arises in algebra, where cyclic words are conjugacy classes, and in the termination literature as a bridge towards graph transformation \cite{ZantemaKonigBruggink2014,SabelZantema2015,SabelZantema2017}. Nearby undecidability results concern different problems. Narendran and Otto exhibited a finite length-reducing complete string system with undecidable cyclic equality, an algebraic question modulo the induced Thue congruence rather than confluence of the directed relation on rotation classes \cite{NarendranOtto1986}. Otto proved confluence on a given congruence class undecidable even for finite length-reducing string systems \cite{Otto1987}. The closest result specifically on cyclic confluence reduces it to a bounded check under a strongly confluent base system and conjugacy-compatibility hypotheses \cite{DiekertDuncanMyasnikov2012}. These hypotheses do not cover arbitrary terminating cycle systems, and we are not aware of an earlier decidability or undecidability result for confluence on the unrestricted class. On the string side, known two-letter encodings for Thue congruences address undirected equivalence rather than preservation of a directed terminating relation \cite{MadlenerOtto1997}. The compiler of Section~\ref{sec:compiler} needs exactly what such encodings lose, and the difficulty lies in the malformed cycles.

\section{Preliminaries}\label{sec:preliminaries}
We review cyclic words and cycle rewriting, reduction and confluence, certificates and codes, and the machine model behind the reduction.

\subsection{Cyclic Words and Cycle Rewriting}
Let $A$ be a finite alphabet, $A^*$ the words over $A$, and $\eps$ the empty word. Words $u,v\in A^*$ are \emph{rotationally equivalent}, written $u\rot v$, if $u=xy$ and $v=yx$ for some $x,y\in A^*$. The cyclic word represented by $u$ is the rotation class $\cyc{u}=\{v\in A^*:v\rot u\}$; we call a cyclic word a \emph{cycle} for short. Rotation does not include reversal.

A finite string-rewriting system is a finite set $R\subseteq A^+\times A^*$, with rules written $\ell\to r$. Its induced relation on cyclic words is $\cyc{u}\cstep_R\cyc{v}$ when, for some rule $\ell\to r$ and word $z$,
\[
  u\rot \ell z,
  \qquad
  v\rot rz.
\]
Equivalently, a left-hand side is matched along a simple directed arc. Thus $|\ell|\le |u|$ is required: a match may cross the displayed end of a representative, but it may not wind around a shorter cycle. We call this requirement the \emph{no-winding convention}; it is used throughout. For example, $ab\to b$ applies to $\cyc{ba}$, since $ba\rot ab$, giving $\cyc{b}$. Over an alphabet containing two distinct letters, rotation equivalence is not a congruence on $A^*$: for distinct $a,b$, we have $ab\rot ba$ but $abab\not\rot baab$. Thus, in general, the induced relation is not string rewriting modulo an equational theory, and completion methods for rewriting modulo \cite{JouannaudKirchner1986} do not apply directly.

\begin{example}\label{ex:context}
Let $R=\{aba\to b\}$ over $A=\{a,b\}$. The only proper self-overlap of the left-hand side is $bba\leftarrow ababa\to abb$, and its two targets are rotations of each other, so the cyclic peak joins in zero steps. Embedding the same overlap in the cyclic context $ab$ destroys this: $\cyc{ababaab}$ has the two descendants $\cyc{bbaab}$ and $\cyc{abbab}$. The first is a normal form, the second reduces only to $\cyc{bbb}$. Joinability of an overlap is therefore not stable under cyclic context.
\end{example}

\subsection{Reduction and Confluence}
We use the standard notions of abstract rewriting \cite{BaaderNipkow1998}. Let $\to$ be a binary relation on cyclic words and write $\to^*$ for its reflexive transitive closure. We say that:
\begin{enumerate}
  \item a \emph{derivation} is a sequence of $\to$-steps, and a \emph{descendant} of $C$ is any $D$ with $C\to^*D$;
  \item a \emph{normal form} is a cycle to which no rule applies;
  \item $C$ and $D$ are \emph{joinable} if they have a common descendant;
  \item $\to$ is \emph{terminating} if there is no infinite derivation;
  \item $\to$ is \emph{locally confluent} if the branches of every one-step peak are joinable, and \emph{confluent} if, whenever $C\to^*D$ and $C\to^*E$, the cycles $D$ and $E$ are joinable.
\end{enumerate}

\begin{theorem}[Newman's lemma \cite{Newman1942}]\label{thm:newman}
A terminating relation is confluent if and only if it is locally confluent.
\end{theorem}

\begin{lemma}[Unique normal forms]\label{lem:unique-nf}
A terminating relation is confluent if and only if every cycle has a unique normal form.
\end{lemma}

\begin{proof}
Under termination every cycle has at least one normal form. If the relation is confluent, any two normal forms of the same cycle are joinable and hence equal. Conversely, any two descendants of a cycle reduce to normal forms, which both equal its unique one, so all branches are joinable.
\end{proof}

\subsection{Certificates and Codes}
\begin{definition}[Certificates and nonerasing systems]
A \emph{positive additive weight certificate} for $R$ is a function $\omega:A\to\N_{>0}$, extended additively to words, such that $\omega(\ell)>\omega(r)$ for every rule $\ell\to r$. The system is \emph{nonerasing} if every right-hand side is nonempty.
\end{definition}

A system is \emph{strictly length reducing} if $|r|<|\ell|$ for every rule. This is the special case in which every letter has weight one. A certificate is finite data verified by one weight comparison per rule, so certified termination is evident from a system's presentation rather than proved by search.

\begin{proposition}\label{prop:certificate}
If $R$ has a positive additive weight certificate, then its cycle relation terminates, no derivation has more steps than its source's weight, and the descendants of a fixed cycle form a finite computable set.
\end{proposition}

\begin{proof}
Rotation preserves total weight and every step strictly decreases it, so weights bound derivation lengths. Every descendant has weight, and hence length, at most the weight of the source cycle, so the descendant set is finite and computed by closing under the one-step relation.
\end{proof}

For a fixed finite alphabet $\Omega$, we use a standard total coding of finite rule lists; whether a number codes a well-formed system is decidable, and a decision problem about systems becomes a set of codes. A set of codes is in $\Pi^0_1$ when its complement is computably enumerable. Equivalently, $\Pi^0_1$ consists of the sets of the form $\{x:\forall y\,R(x,y)\}$ with $R$ decidable, the first universal level of the arithmetical hierarchy \cite{Soare2016}; its dual existential level $\Sigma^0_1$ consists of the computably enumerable sets, and completeness is defined dually. Membership in a $\Pi^0_1$ set can therefore fail only for a finite reason, a single checkable counterexample $y$, while membership itself asserts that infinitely many checks all pass. Confluence of a terminating system has exactly this shape: a nonjoinable peak is a finite refutation, while confluence asserts joinability of every peak from every cycle. A $\Pi^0_1$ set is $\Pi^0_1$-complete when every $\Pi^0_1$ set reduces to it under a total computable many-one reduction, and $\Pi^0_1$-complete sets are undecidable. We write $(W_e)_{e\in\N}$ for a standard enumeration of the computably enumerable subsets of $\N$; every $\Pi^0_1$ set is then the complement of some $W_e$, the form in which hardness enters our reductions.

\subsection{Shrinking Local Machines}
\begin{definition}[Deterministic shrinking local machine]\label{def:local-machine}
A deterministic shrinking local machine is a tuple $M=(D,\Sigma,Q,q_0,F,k,\lambda,\rho,\mathcal T)$, where $D$ is a finite data alphabet, $\Sigma\subseteq D$ is the input alphabet, $Q$ is a finite state set, $q_0\in Q$, $F\subseteq Q$, $k\ge1$, $\lambda,\rho\in D^*$, and $\mathcal T$ is a finite partial deterministic table of rules
\[
  \alpha q\beta\longrightarrow u q'v,
  \qquad
  \alpha,\beta\in D^k,
  \quad u,v\in D^*,
  \quad |u|+|v|<2k.
\]
A configuration is a word in $D^*QD^*$ containing exactly one state symbol. A table rule applies by replacing an occurrence of the factor $\alpha q\beta$ in a configuration by $u q'v$. There is at most one table entry for each triple $(\alpha,q,\beta)$, and final states have no outgoing entry. On input $w\in\Sigma^*$ the initial configuration is $\lambda q_0w\rho$. The input is accepted when the unique table computation reaches a final state; a nonfinal stuck configuration rejects. The \emph{accepted language} $L(M)\subseteq\Sigma^*$ is the set of accepted inputs.
\end{definition}

Every table step strictly shortens the configuration: its left side has $2k$ data symbols and one state, its right side at most $2k-1$ and one state. In particular the computation on input $w$ stops within $|\lambda w\rho|$ steps, so acceptance of a given input is decidable; it is emptiness of $L(M)$ that will prove undecidable.

\section{Overview of the Construction}\label{sec:overview}

\paragraph{The three transformations.}
Each equivalence in the chain of Section~\ref{sec:introduction} is realized by one construction with one governing invariant. Theorem~\ref{thm:verifier-source} turns halting of a two-counter program into nonemptiness of a deterministic shrinking local machine whose accepted inputs certify themselves. An accepted word carries a core that replays a genuine halting run, with explicit fuel paying for every head movement. Theorem~\ref{thm:weighted-fork} wraps the machine into a weighted cycle system with a single intentional fork and unique normal forms everywhere else. Theorem~\ref{thm:binary-transfer} recodes any finite nonerasing weight-decreasing cycle system over a nonempty alphabet in binary so that no rule turns a valid encoding invalid or an invalid one valid, and finitely many cleanup rules drain every reducible malformed cycle to one error normal form. Nonconfluence therefore always has a one-cycle witness, while confluence must be certified on all cycles at once, so the burden of every stage falls mainly on inputs it was never meant to receive.

\paragraph{Obstructions to a direct reduction.}
Four constraints stand between a Turing-complete model and the target class. Every rule must strictly reduce length, so divergence can encode nothing and a length-preserving machine step cannot be transliterated. Whatever the target computes must be paid for by material already present on the starting cycle. Right-hand sides must be nonempty, so garbage cannot be erased outright but has to be collected into a stable nonempty normal form. The alphabet is fixed to $\B$, so no reserved symbol can flag errors or delimit segments, and unused bit patterns act as inert separators exposing overlaps that no encoding intended. Finally, correct behaviour of the intended launch cycles decides nothing by itself. A single unjoinable peak on any malformed cycle destroys confluence, so every unintended input must be controlled rather than merely avoided. Each stage absorbs the constraint it is best placed to meet. The machine does the resource accounting, the wrapper the branch control, and the binary code the error draining.

\paragraph{The single fork and the machine model.}
All nonjoinable branching of the final system originates in one rule pair, the launch fork, which commits a cycle to one of two branch bits. Both branches run the same deterministic simulation and the bit survives only into the acceptance outcome. Every other component, from input decoding through rejection and error sweeping, has unique normal forms by design. The format of Definition~\ref{def:local-machine} is chosen for this embedding. A table rule $\alpha q\beta\to uq'v$ with $|u|+|v|<2k$ is itself a strictly length-reducing local rewriting rule, so the wrapper adopts the table up to branch-tagging its state symbols, one cycle rule per entry and branch, and inherits its determinism, at most one entry per window and none from final states, as unique normal forms of the simulation phase. The fixed radius bounds the boundary and collision windows the wrapper must guard, and the fixed five-symbol data alphabet and padding fix the shapes of all wrapper rules, the instance entering only through the state set and the table, so one wrapper scheme serves the entire uniform family of Theorem~\ref{thm:verifier-source}. What the format forbids, the input supplies. A shrinking configuration cannot support an unbounded run, so accepted words carry their operation budget and movement fuel in advance, and nonemptiness quantifies over all such provisionings, a halting run being caught by some exactly provisioned witness because it needs only finitely many resources.

\section{A Fixed-Radius Shrinking Verifier}\label{sec:verifier}
This section proves the following compilation lemma, from which Theorem~\ref{thm:verifier-source} follows. Two-counter programs, whose instructions increment a counter or combine a zero test with a decrement, are recalled below.

\begin{lemma}[Fuelled two-counter verification]\label{lem:fuel-compile}
Let $P$ be a finite deterministic two-counter program with a designated initial configuration and halting control $h$. Effectively from $P$ one can construct a deterministic shrinking local machine $V_P$ such that
\[
  L(V_P)\ne\varnothing
  \quad\Longleftrightarrow\quad
  P\text{ reaches }h\text{ from its initial configuration}.
\]
The machine uses the fixed data and input alphabet $D=\Sigma=\{0,1,f,\mathtt L,\mathtt R\}$, window radius $k=2$, fixed padding words $\lambda=\mathtt{LL}$ and $\rho=\mathtt{RR}$, and a nonfinal start state. Only its finite state set and transition table depend on $P$ and its initial configuration.
\end{lemma}

We regard the initial configuration as part of the program instance and write the resulting machine simply as $V_P$. The construction has three layers: a physical representation of verifier inputs with explicit movement fuel; a finite collection of macros, the entry boot pass and the pay, seek, update, and return sequence that realizes each program instruction; and correctness lemmas that assemble the completed macro executions into whole computations.

\subsection{Counter Programs}

A deterministic two-counter program has finitely many controls $p$, among them a designated initial control $p_{\mathrm{init}}$ and the halting control $h$, two nonnegative counters, increment instructions $p:\ C_i:=C_i+1;\ \mathrm{goto}\ p'$, and conditional instructions $p:\ \mathrm{if}\ C_i=0\ \mathrm{goto}\ p_0$ $\mathrm{else}\ C_i:=C_i-1;\ \mathrm{goto}\ p_1$. The halting control has no outgoing instruction, and every other control carries exactly one. The designated initial configuration is $(p_{\mathrm{init}},v_0,v_1)$, where $(v_0,v_1)$ is the designated initial valuation.

\subsection{Physical Witnesses}

A witness for $V_P$ supplies three unary resources: an operation budget $B$ and counter capacities $R_0,R_1$. Its \emph{logical form} is the word
\[
  1^{B}\;
  \mathtt{L}\,1^{v_0}0^{R_0-v_0}\,\mathtt{R}\;
  \mathtt{L}\,1^{v_1}0^{R_1-v_1}\,\mathtt{R},
\]
where $v_i$ is the designated initial value of counter $i$, the resources satisfy $B\ge0$ and $v_i\le R_i$, and $\mathtt{L},\mathtt{R}$ are the block boundary markers. A \emph{physical witness} is obtained from the logical form by inserting a nonempty fuel reservoir $f^{m}$ at every boundary between consecutive logical cells, including block boundaries and both ends. Displays of logical forms suppress the interleaved fuel.

In particular, every budget token is physically flanked by reservoirs, so deleting it merges them. Counter $i$'s current value $m_i$ is stored as the used prefix $1^{m_i}$ of its block: an increment changes the leftmost reserve $0$ to $1$, a positive decrement changes the rightmost used $1$ to $0$, and a zero test checks that the used zone is empty. Every simulated instruction also consumes one budget token.

For example, with budget $B=2$, capacities $R_0=2$ and $R_1=1$, and initial values $v_0=1$ and $v_1=0$, the initial configuration on one physical witness is
\[
  \mathtt{LL}\;q_0\;
  f^{\mu_1}\,1\,f^{\mu_2}\,1\,f^{\mu_3}\,
  \mathtt{L}\,f^{\mu_4}\,1\,f^{\mu_5}\,0\,f^{\mu_6}\,\mathtt{R}\,
  f^{\mu_7}\,
  \mathtt{L}\,f^{\mu_8}\,0\,f^{\mu_9}\,\mathtt{R}\,
  f^{\mu_{10}}\;\mathtt{RR}
\]
with every reservoir size $\mu_j\ge1$.

The computation begins with a deterministic \emph{boot pass} that scans right from the left padding, turns at the first $\mathtt{RR}$ beyond the witness fields, and returns to the left-hand hub. Since $\Sigma=D$, a radius-two window cannot tell whether the $\mathtt{RR}$ it turns at is the genuine padding $\rho$ or an input factor. A successful boot therefore certifies a physical-witness \emph{core} extending from $\lambda$ to that occurrence, its right \emph{fence}; every later macro operates strictly left of the fence, so a suffix beyond it is never inspected and is semantically inert.

Boot states form their own namespace, whose only exit is the parking step into the initial hub, or into the final state when $p_{\mathrm{init}}=h$; return states later target arbitrary nonfinal hubs, including the initial one, but none is reachable before that exit. Correctly shaped but wrongly provisioned cores may also get stuck, which is harmless: completeness needs only one exactly provisioned witness.

A failed reverse crossing may consume the last cell of the $\lambda$-adjacent reservoir, but it then sticks in a nonfinal configuration; every boot or return execution that reaches its parking exit leaves that reservoir nonempty, and parking never consumes its final cell, as the fuel ledger below records. The boot entry accepts only fuel, and the core's interior $\mathtt{L}$ markers are never adjacent, so on every pass that reaches the return, $\mathtt{LL}$ occurs to the left of the state only at the true left end and parking directly after $\lambda$ is unambiguous.

A \emph{candidate core} is the data of arbitrary content strictly between $\lambda$ and a chosen occurrence of $\mathtt{RR}$, its fence. For $b\ge0$ and $0\le m_i\le R_i$, a \emph{runtime logical core} is the word $1^{b}\,\mathtt L\,1^{m_0}0^{R_0-m_0}\,\mathtt R\,\mathtt L\,1^{m_1}0^{R_1-m_1}\,\mathtt R$, its budget and counter fields then \emph{legal}, and a candidate core is \emph{witness-shaped} when it carries a nonempty fuel reservoir at every boundary, interleaved as above, and deleting its fuel leaves a runtime logical core. The logical form of a witness is the runtime logical core with $b=B$ and $m_i=v_i$, so a physical witness is witness-shaped, and a completed boot certifies its candidate core witness-shaped with $m_i=v_i$. Since a witness-shaped core carries fuel at every boundary, no two of its logical cells are adjacent and $\mathtt{RR}$ occurs nowhere within it, so at most one candidate core is witness-shaped and its fence is the first $\mathtt{RR}$ after $\lambda$, the occurrence a completed boot turns at.

\subsection{Macros and Contracts}

The machine $V_P$ is organized as a finite collection of \emph{macros}, each expanded into a finite set of directional phase states with a private namespace. Table~\ref{tab:macros} lists their contracts. The boot pass is the entry macro, and a return that targets $h$ parks directly in the designated final state, which has no outgoing table rule.

\begin{table}[!ht]
\caption{Macro contracts of $V_P$.}
\label{tab:macros}
\centering
\small
\setlength{\tabcolsep}{3.5pt}
\renewcommand{\arraystretch}{1.0}
\begin{tabularx}{\textwidth}{@{}l >{\raggedright\arraybackslash}p{0.34\textwidth} >{\raggedright\arraybackslash}X@{}}
\hline
Macro & Entry condition & Successful effect \\
\hline
Boot & start state; witness-shaped core with $m_i=v_i$ for $i=0,1$ & certifies the core, parks at the hub (final state if $p_{\mathrm{init}}=h$) \\
Pay & at the hub; budget nonempty & deletes one budget token \\
Seek $i$ & at the seek entry $c_p$ & moves the control to block $i$ \\
Increment $i$ & reserve $0$ at the used edge & $m_i\mapsto m_i+1$, branch $p'$ \\
Conditional $i$ & used zone empty, or used $1$ at the used edge & empty used zone: branch $p_0$; else $m_i\mapsto m_i-1$, branch $p_1$ \\
Return & at a return entry & re-enters the target hub, unchanged but for consumed fuel; for $h$, the final state \\
\hline
\end{tabularx}
\end{table}

Failure of a semantic guard, malformed layout, or exhausted fuel leaves the corresponding phase stuck. An instruction at control $p$ is realized as the macro sequence pay, seek $i$, update, return, where the update macro of an increment rewrites at the used edge and branches to $p'$, that of a conditional dispatches on the cell at the used edge between the decrement continuation towards $p_1$ and the zero continuation towards $p_0$, and the return re-enters the hub at the branch target. A conditional thus compiles to one update macro with a single entry and two exits, its turnaround internal to the positive branch. If the initial control is $h$, the boot exit parks in the final state likewise. The exit state of each macro is the entry state of the next, which the run decomposition below uses.

\subsection{The Command Table}\label{sec:command-table}

\paragraph{States.}
A nonfinal state of $V_P$ is a tuple consisting of a macro, one of its finitely many phases, a direction, and a boundary class drawn from the finite set of field tags. The hub states are the exit states of the return phases, one per program control other than $h$; a return targeting $h$ exits into the one designated final state instead. Every component ranges over a set that is finite and computable from $P$, so the state set is finite and computable.

\paragraph{Transition families.}
Every table entry is an instance of one of three right-moving families or their left-moving mirrors. With $x,y,c$ ranging over cells and $f$ denoting literal fuel, the right-moving forms are
\begin{align*}
  &\text{(T1) move:} & x\,y\,q\,c\,f&\longrightarrow x\,y\,c\,q' &&\text{cross $c$, consuming the fuel beyond it,}\\
  &\text{(T2) rewrite:} & x\,y\,q\,c\,f&\longrightarrow x\,y\,q'\,c' &&\text{rewrite $c$ to $c'$, consuming the fuel,}\\
  &\text{(T3) delete:} & x\,y\,q\,f\,c&\longrightarrow x\,y\,q'\,f &&\text{delete $c$, preserving the fuel before it.}
\end{align*}
The left-moving mirrors are $f\,c\,q\,y\,x\to q'\,c\,y\,x$, $f\,c\,q\,y\,x\to c'\,q'\,y\,x$, and $c\,f\,q\,y\,x\to f\,q'\,y\,x$, read symmetrically. Each form has one state and four data cells on the left and at most three data cells on the right, so it fits the $k=2$ table format and shrinks. The families are schemata, not simultaneously generated alternatives: \textup{(T1)} and \textup{(T2)} share their left-hand shape outright, and \textup{(T3)} joins them when the scanned pair is $ff$, so a window alone selects no family; each node below carries at most one command per window, and that command names its family and successor.

\paragraph{Compilation.}
Each macro is specified by a finite deterministic \emph{command graph}: its nodes are the macro's states, and each node carries finitely many guarded commands, one per expected window, each emitting a single instance of one family or mirror and naming the successor node. Each guard is an exact local key in $D^2\times Q\times D^2$, determined by the node and one window; the guards of a node are pairwise disjoint, and the regular field structure of the physical witness, its fuel runs, token runs, and marker-delimited blocks, is recognized by sequences of such keys threaded through the phase states, generated by the schemata of Appendix~\ref{app:catalogue}, so every graph is finite and computable from $P$.

Expanding all commands of all macros yields the transition table of $V_P$, and Table~\ref{tab:compiler} specifies the graphs by phase kind.

\begin{table}[!ht]
\caption{Command graphs of $V_P$ by phase kind.}
\label{tab:compiler}
\centering
\small
\setlength{\tabcolsep}{3pt}
\renewcommand{\arraystretch}{1.0}
\begin{tabularx}{\textwidth}{@{}l >{\raggedright\arraybackslash}p{0.2\textwidth} >{\raggedright\arraybackslash}p{0.24\textwidth} >{\raggedright\arraybackslash}p{0.16\textwidth} >{\raggedright\arraybackslash}X@{}}
\hline
Phase & State memory & Accepted cases & Families & Successor \\
\hline
Boot scan & field index, position in $1^{v_i}$ & at $q_0$ fuel only; the expected fuel, token, marker, reserve cell; the fence & \textup{(T1)}; turn mirrored \textup{(T3)} & next scan phase; at the fence, boot return \\
Boot return & direction only & mirrored field cells; two fuels at $\lambda$ & mirrored \textup{(T1)}; parking \textup{(T3)} & $s_{p_{\mathrm{init}}}$, or the final state if $p_{\mathrm{init}}=h$ \\
Pay & control $p$; entry/active phase & two fuels; fuel then token & \textup{(T3)} & pay-active $t_p$; on the token, seek entry $c_p$ \\
Seek $i$ & instruction, target block, field class, $\mathtt L$s passed & fuel; budget $1$; counter $0/1$; $\mathtt L/\mathtt R$ per field class & \textup{(T1)} & update entry inside block $i$ \\
Increment $i$ & block $i$; branch target $p'$ & fuel; used $1$s; first reserve $0$; none at $\mathtt R$ & \textup{(T1)}; one \textup{(T2)} & return towards $p'$ \\
Conditional $i$ & block $i$; targets $p_0,p_1$; seen-$1$ bit & fuel; used $1$ sets the bit; reserve $0$ or $\mathtt R$ dispatches & \textup{(T1)} & return towards $p_0$, or turnaround \\
Turnaround & block $i$; branch target $p_1$ & the dispatch cell $0$ or $\mathtt R$; mirrored fuel; the rightmost used $1$ & mirrored \textup{(T1)}; one mirrored \textup{(T2)} & return towards $p_1$ \\
Return & target control, field class & mirrored field grammar; two fuels at $\lambda$ & mirrored \textup{(T1)}; parking \textup{(T3)} & hub $s_{p'}$, or the final state for $h$ \\
\hline
\end{tabularx}
\end{table}

Each row abbreviates finitely many exact keys in $D^2\times Q\times D^2$, one per expected window per listed case, the context pair ranging over the admissible sets defined below. Cases within one state are pairwise disjoint; scan, pay-active, dispatch, turn, update, and parking states are distinct; keys not listed have no table entry, so unexpected windows stick.

Every traversing base state $q$ has an unflagged copy $q$ and a flagged copy $\hat q$. The unflagged fuel-pair command loops from $q$ to $q$; the flagged fuel-pair command, generated once for every flagged copy that receives keys, goes from $\hat q$ to $q$ and clears the flag; every same-direction crossing of a logical cell targets the flagged copy of its successor; a handoff that reverses direction targets the unflagged copy; the pay states are unflagged throughout.

The update and return phases use increment states $a_{i,p'}$, carrying the block index and the branch target, conditional states $b_i^{\beta}$, carrying the block index, the targets $p_0,p_1$, and the seen-$1$ bit $\beta$, turnaround states $u_{i,p_1}$, carrying the block index and the branch target, and return states $r^{\phi}_p$ with target control $p$ and field class $\phi$, the class carrying its block index where one exists; here and below, printed instances abbreviate $a_{i,p'}$, $b_i^{\beta}$, $u_{i,p_1}$ as $a_{p'}$, $b^{\beta}$, $u_{p_1}$, suppressing the block index and the conditional's targets, so increments of the two counters towards one target are distinct states.

Each macro-entry state is a distinguished \emph{entry copy} of its first phase state: it carries that state's accepted cases, flagged or unflagged according to the command that enters it, with the fuel-only rule prevailing at a flagged entry copy; its commands take the successors its bound state's clauses name, and no command internal to a macro targets an entry copy. The boot entry $q_0$ is further restricted in Appendix~\ref{app:catalogue}: its sole command opens the direct traversal of the leading reservoir.

The seen-$1$ bit is the finite-control record that a used $1$ was crossed, and a return state carries its target program control, so the phase-to-exit correspondence of Table~\ref{tab:macros} is realized literally. The command graphs are generated by the schemata of Appendix~\ref{app:catalogue}, one clause per accepted case.

A command's exact key records both sides of the state: the scanned pair says what the phase faces next, and the rear pair records where it has been. For a phase state $q$ let $\mathrm{lead}(q)\subseteq D^2$ collect the scanned pairs its clauses list and $\mathrm{trail}(q)\subseteq D^2$ the rear pairs its clauses are generated against, each pair written in tape order, so the cell adjacent to the state stands last in a right-moving state's rear pair and first in a left-moving one's. A phase state's \emph{positional invariant} fixes its position, direction, and field class within the core, at boot within the candidate core up to the chosen fence; Appendix~\ref{app:catalogue} records it phase by phase beside each macro's clauses, its per-state sets listing the trail pairs outright. It carries no semantic guard, counter value, or seen-$1$ assertion; those are data conditions, established with the macro contracts below. The candidate keys are
\[
  \mathrm{Adm}_0(q)=
  \begin{cases}
    \mathrm{trail}(q)\times\{q\}\times\mathrm{lead}(q), & q\ \text{right-moving},\\
    \mathrm{lead}(q)\times\{q\}\times\mathrm{trail}(q), & q\ \text{left-moving},
  \end{cases}
\]
each key listing the pair left of the state first, so the scanned pair stands first exactly at a left-moving state. The admissible set $\mathrm{Adm}(q)$ is obtained from $\mathrm{Adm}_0(q)$ by the sole per-clause restrictions of the product: the keys of the three boundary deletions, the turn at the fence and the parkings at $\lambda$, boot and return alike, are retained only against the rear pair $ff$, so a last-survivor rear sticks at an inadmissible window. The table carries exactly one command per admissible key: the generator reads the scanned pair alone and emits the family instance and successor of the matching clause, the rear pair carried unchanged. Every flagged copy that receives keys has $\mathrm{lead}(\hat q)=\{ff\}$, one flag-clearing key per rear pair.

An entry copy carries its own sets, those of its bound state except where its clause restricts them, as at $q_0$.

The state set is the full constructor product over the program's controls and instructions, together with the designated final state, but the table is generated inductively from the start state and the hubs: every other nonfinal tuple receives its keys exactly when it has a \emph{creation site}, an already generated command naming it as successor, while the final state, entered by its parkings, has no outgoing rule; a tuple without a creation site, such as the flagged copy $\hat b^{0}$, which no command names since crossing a used $1$ targets $\hat b^{1}$ regardless of the current seen-$1$ bit, has an empty trail set, receives no keys, and is inert. Equivalently, the tuples receiving keys form the least set of nonfinal tuples containing the start state and the hubs and closed under the nonfinal successors their commands name. The return classes generated at a target $p$ are thus exactly those the walk of Appendix~\ref{app:catalogue} reaches in zero or more steps from the classes its creation sites bind: every walk step stays in its field class or moves to one further left, and a single step joins each consecutive pair of classes, so these are the bound classes together with every class to their left. This closure is a finite computation on the command graphs, so the table is determined with no reachability analysis of machine configurations.

\begin{lemma}[Generated table]\label{lem:generated-table}
The command graphs expand effectively to a finite deterministic shrinking local table of radius two. Every rule has four data cells on the left and at most three on the right, the final state has no outgoing rule, macro namespaces are disjoint, and the only rule leaving the boot namespace is its parking exit. No command internal to a macro targets a macro-entry state: macro-entry states occur as successors only at the named macro handoffs and parking exits, and the designated final state occurs only as the parking target associated with the halting control $h$.
\end{lemma}

\begin{proof}
Finiteness, radius two, and shrinking follow directly from Table~\ref{tab:compiler} and the three transition families with their mirrors. Determinism holds because each admissible key of the restricted piecewise product generates exactly one entry, by the catalogue in Appendix~\ref{app:catalogue}, and distinct clauses of one state have distinct scanned pairs; the coinciding family shapes, \textup{(T1)} with \textup{(T2)} everywhere and both with \textup{(T3)} at the scanned pair $ff$, are harmless because the generated command, not the window, names the family. The namespace and final-state claims follow from the listed successors.

For the entry discipline, the successor column names a macro-entry state only at the handoffs, pay to the seek entry $c_p$, seek to its update entry, and the update phases to their return entries, and at the parking exits of boot and return, which target hubs; every other successor, a scan phase, the boot return, the pay-active state $t_p$, or the turnaround, is internal to its macro's namespace, and the schema fuel loops and flag clearings target interior copies, never entry copies. The final state appears only in the two parking rows, in place of the hub exactly when the target control is $h$.
\end{proof}

\begin{lemma}[Exhaustive generation]\label{lem:exhaustive}
At any phase state, if the window at the state is an admissible key, then exactly one command applies; if it is not, including when inadmissibility stems from a failed semantic guard, exhausted fuel, or a malformed field cell, no command applies. If moreover the state satisfies its positional invariant within a witness-shaped core, or occurs along the boot pass from $q_0$ within its candidate core, then the step establishes the successor's positional invariant, its position, direction, and field class within the core, though the successor's own window need not be admissible.
\end{lemma}

\begin{proof}
The table carries one command per key of $\mathrm{Adm}(q)$, and distinct clauses of one state have distinct scanned pairs, giving existence and uniqueness on admissible windows. For transport within a witness-shaped core, a fuel-pair step keeps the class and shortens the residue, a crossing advances the class, flagging its successor when the direction is kept, and a rewrite or deletion changes exactly the cell its clause names, so every command stands its named successor at the position, direction, and field class the successor's invariant prescribes, by inspection of the clauses; a parking fires only against a scanned $\mathtt{LL}$, and a witness-shaped core, no two of whose logical cells are adjacent, presents that factor at $\lambda$ alone, so the hub or final state stands directly after $\lambda$. Along the boot pass, whose candidate core is not yet certified, reachability replaces placement: by induction from $q_0$'s opening of the leading reservoir, the fuel-erased data behind a scan state is the witness-grammar prefix of Appendix~\ref{app:catalogue}, the turn stands $q_{\mathrm{ret}}$ left of the fence it fires at, and the retrace crosses only cells the scan validated, whose interior $\mathtt L$ markers are never adjacent, while the leading reservoir, nonempty by the opening key of $q_0$, keeps a survivor until the retrace crosses the first logical cell, so the parking factor $\mathtt{LL}$ that $q_{\mathrm{ret}}$ meets against a rear $ff$ is $\lambda$ itself; an arbitrarily placed $q_{\mathrm{ret}}$ could instead park at an interior $\mathtt{LL}$ of its candidate core, which is what confines the boot claim to the pass. This is all a step establishes: a consumption that empties the reservoir at the successor's window leaves the successor correctly placed over a window outside its sets, one of the stuck windows Appendix~\ref{app:catalogue} records phase by phase.

A window outside $\mathrm{Adm}(q)$ carries no command by construction, whether it fails on the scanned side, a failed semantic guard, an empty budget or a full block among them, exhausted fuel, or a malformed field cell, or on the rear side, as after a consumption that emptied the following reservoir; the judgement is at the window alone.
\end{proof}

\subsection{Fuel Ledger and Allocation}

\begin{lemma}[Traversal]\label{lem:traversal}
A \emph{seam} traversal, entered by crossing the preceding logical cell, turns a reservoir of $2m+1$ cells, $m\ge1$, into its $m$ alternating survivors, the state exiting unflagged at the next logical cell with the survivors behind it; a \emph{direct} traversal, by a phase starting adjacent to the reservoir, pairs its cells from the first and turns $2m$ into $m$, again for $m\ge1$ and exiting likewise. A \emph{single consumption}, one scheduled \textup{(T2)} or \textup{(T3)} step or one \textup{(T1)} crossing followed by no traversal, takes $m+1$ to $m$, and a \emph{payment merge}, in which the payment macro contracts its leading reservoir to one cell whose token deletion merges it into the following reservoir, takes positive sizes $(a,b)$ to $b+1$. Sizes of the wrong parity, or too small, stick strictly inside the traversal, at a flagged state or a wrong-parity window, so successful traversal sizes are exact rather than merely large.
\end{lemma}

\begin{proof}
A movement phase traverses a reservoir with \textup{(T1)} steps that cross its cells pairwise: each interior step crosses one fuel cell, which survives behind the state, and consumes the next. The event effects are exact. A seam entry crossing consumes the reservoir's first cell and flags the state, the flagged fuel-pair step clears the flag, and the unflagged pairs then halve the remaining $2m$ cells, so the exit crossing departs unflagged with the alternating survivors behind it; the direct pairing from the first cell halves $2m$ likewise, and a single consumption and the payment deletions remove exactly the cells their instances name, the contracted leading cell joining the following reservoir at the token deletion. For the failure clause, entry flags the state, whose only windows continue with fuel, so a seam-entered even reservoir blocks at a singleton and a singleton seam reservoir blocks directly after its entry crossing.
\end{proof}

Every parking and payment command demands two visible fuel cells or a token beyond the fuel, while a reverse crossing of the first logical cell may consume the last cell of the $\lambda$-adjacent reservoir and then sticks flagged at $\lambda$; hence every boot or return execution that reaches its parking exit leaves that reservoir nonempty, and parking never consumes its final cell.

\begin{lemma}[Effective fuel allocation]\label{lem:fuel-alloc}
A finite intended run, a sequence of macros from boot through its final parking, induces one finite chronological trace of seam traversals, direct traversals, single consumptions, and payment merges over its evolving reservoirs; a seam traversal may straddle a handoff, its entry crossing in one macro and its completion in the next. Some effectively computable assignment of positive initial reservoir sizes enables every event of the trace and leaves every prescribed survivor nonempty. Applying the event effects along the trace, the assignment determines the reservoir sizes after every prefix of the trace; these intermediate sizes are positive, and each event, the completion of a straddling seam included, is enabled from the sizes current at its own point.
\end{lemma}

\begin{proof}
Assign positive requirements to the terminal reservoirs and invert the events in reverse chronological order. Traversals and single consumptions invert exactly, by Lemma~\ref{lem:traversal}, through $m\mapsto2m+1$, $m\mapsto2m$, and $m\mapsto m+1$.

At a payment merge the required merged size $N$ is at least two, since payment is immediately followed by the seek macro's direct traversal of the merged reservoir, whose preimage is even and positive; split $N$ as $1+(N-1)$, assign $1$ to the leading reservoir, whose payment completes from every positive size, and $N-1$ to the following one, and continue on both predecessors. Reverse induction over the finite trace, itself read off the macro paths, keeps every requirement positive. Read forward, the event effects are exact, so the size after each prefix equals the requirement computed at that point of the reverse induction; hence the intermediate sizes are positive and each event finds exactly its enabling size, a straddling seam receiving one requirement at its single trace event, shared by its entry crossing and its completion.
\end{proof}

A fully worked instruction, tracing all six phases with its ledger table and reverse allocation, appears in Appendix~\ref{app:calculations}.

\subsection{Macro Correctness}

A \emph{stop} is a positive-time visit of a macro-entry state, a positive-time visit of a designated final state, or the exhaustion of the run at a stuck configuration, possibly at time zero. A \emph{schedule} is the segment of a run from a macro-entry state until its first stop; it is \emph{completed} exactly when the stop is one of the two visits. In particular, arrival at a macro-entry state that is itself stuck completes the arriving schedule; the next schedule from that configuration exhausts at time zero and is incomplete. For a macro-entry configuration, its \emph{macro-local execution} is the unique maximal execution truncated at its first stop.

A \emph{certified macro boundary} is a macro-entry configuration of a non-boot macro whose data up to the fence is a witness-shaped core with legal budget and counter fields and whose state is its macro's entry state at that macro's start position. An \emph{admissible macro boundary} is either the boot entry state on an input with a witness-shaped candidate core or a certified macro boundary. Entry conditions are static properties of the boundary configuration; that boundaries arise in their macro order along an actual run is supplied by the entry discipline of Lemma~\ref{lem:generated-table}.

\begin{lemma}[Boot certification]\label{lem:boot-cert}
If the boot schedule from a start configuration completes, then the data up to the chosen fence is a witness-shaped core with $m_i=v_i$ and the endpoint is the first certified macro boundary, or the final state when $p_{\mathrm{init}}=h$.
\end{lemma}

\begin{proof}
If the boot schedule completes, every scanned window matched an expected key, so each field up to the fence agrees with the witness grammar, with used prefixes of the designated lengths $v_i$, and the leading reservoir is nonempty by the opening key of $q_0$; since interior $\mathtt L$ markers are never adjacent in such a core, the parking factor $\mathtt{LL}$ occurs only at $\lambda$, and the parking exit enters the initial hub, or the final state when $p_{\mathrm{init}}=h$; in the former case the endpoint is the first certified macro boundary.
\end{proof}

\begin{lemma}[Macro safety]\label{lem:macro-safe}
If a completed schedule starts at a certified macro boundary, then the macro's entry condition held and its fuel-erased endpoint realizes the unique successful effect in Table~\ref{tab:macros}: for an ordinary macro the endpoint encodes the prescribed logical successor configuration and is a certified macro boundary or a noncertified macro-entry configuration at which a reservoir is empty and no table command applies, and for a return targeting $h$ it is the designated final state.
\end{lemma}

\begin{proof}
By Lemmas~\ref{lem:generated-table} and~\ref{lem:exhaustive}, a completed schedule follows its macro's unique guarded path, so the logical fields of the endpoint are those the contract prescribes: the increment rewrites only on a reserve $0$, keeping $m_i\le R_i$, the turnaround rewrites only a used $1$, keeping $m_i\ge0$, payment deletes only a token behind a nonempty budget, and every other phase moves without rewriting. The only certification defect a completed schedule can leave is therefore an empty reservoir; an interior depletion or parity failure sticks strictly inside the macro, at a flagged state, a wrong-parity window, or a survivor-stripped rear, an incomplete schedule, so only the final handoff step can leave one empty.

We inspect the handoffs. Payment contracts the leading reservoir to one cell, and its token deletion merges that cell into the following reservoir, nonempty at a certified boundary, so the seek entry always stands at a nonempty merged reservoir. The return parking is generated only against the rear $ff$ and consumes one of the two visible cells, and a completed return has re-crossed every reservoir it traversed leaving a survivor, so its hub endpoint keeps every reservoir nonempty. These endpoints are certified.

Four final steps can empty the reservoir they consume. If seek's crossing of the opener $\mathtt L$ consumes its sole cell, the flagged update entry faces a logical cell and its lead $\{ff\}$ matches nothing. If the increment's \textup{(T2)} consumes the sole cell beyond the rewritten $0$, the return entry stands over the rear pair $10$ or $1\mathtt R$, excluded from its trail $\{1f\}$. If the zero dispatch consumes the sole cell beyond the dispatch cell, the return entry's rear pair opens with a logical cell, excluded from its survivor-leading trail. If the turnaround's mirrored \textup{(T2)} consumes the sole survivor left of the rewritten $1$, the return entry scans $10$ or $\mathtt L0$, outside its lead. In each case the endpoint is noncertified and no command applies; with a survivor left instead, every reservoir is nonempty and the endpoint is certified.

If a schedule from a certified boundary completes, the entry condition holds, by the boundary definition for the positional conditions and forced by the unique guarded path for the data conditions, and the path identifies the contract; deleting fuel leaves exactly the single logical action, so a nonfinal endpoint encodes the prescribed logical successor and, by the inspection above, is certified, or noncertified with no applicable command, while a parking exit targeting $h$ is the designated final state.
\end{proof}

A certified endpoint may itself be stuck, at an empty budget or before a single visible survivor: certification asserts shape, never progress. The boot schedule, whose start is not yet certified, is treated separately in Lemma~\ref{lem:boot-cert}: its completion certifies the core outright.

Fix a finite run of $P$. A provisioning is \emph{adequate} for it when the budget $B$ exceeds the number of instructions executed and each capacity $R_i$ exceeds every value attained by counter $i$. Fix in addition an \emph{allocation}, one assignment supplied by Lemma~\ref{lem:fuel-alloc} for the run's full event trace, boot pass included. The \emph{intended boundary} at a point of the run is the admissible macro boundary of the pending macro whose fuel-erased core is the runtime logical core with the provisioning's capacities $R_0,R_1$, the remaining budget, and the counter values at that point and whose reservoirs carry the sizes the allocation assigns there, at a flagged entry inside a straddling seam the pre-crossing size less the single crossing consumption; at the run's start it is the boot entry on the physical witness itself. Entry conditions then hold at every intended boundary: positional ones by the boundary shape; boot's $m_i=v_i$ because the start boundary is the physical witness itself; the conditional's disjunction by legality of its counter field; and pay's nonempty budget and increment's reserve $0$ by adequacy.

\begin{lemma}[Allocated realization]\label{lem:alloc-real}
Fix a finite run of $P$ with an adequate provisioning and an allocation. From the intended boundary at any point of the run, the macro-local execution is a completed schedule and its endpoint is the intended boundary at the next point of the run, or the final state at the halting exit: for an ordinary macro, allocation selects the certified branch of Lemma~\ref{lem:macro-safe}.
\end{lemma}

\begin{proof}
The run's macro pass at that point is an execution from the intended boundary: its entry condition holds as noted above, Lemma~\ref{lem:exhaustive} gives the case analysis over the certified field grammar, or over the physical witness at boot, and the reverse induction of Lemma~\ref{lem:fuel-alloc} hands every traversal a reservoir of exactly the size and parity its event requires while keeping every intermediate size positive, so every window along the pass carries a command and the pass ends at the next macro entry or at the final state. By Lemma~\ref{lem:generated-table} each configuration has at most one successor, so the macro-local execution is exactly this pass and is a completed schedule. Its endpoint's fuel-erased core carries the run's data at the next point, unchanged at boot and prescribed by Lemma~\ref{lem:macro-safe} for an ordinary macro, and its reservoirs carry the allocation's sizes at the next point, fully traversed ones by the event effects, untouched ones unchanged, and the one straddled at a flagged entry the pre-crossing size less the single crossing consumption; hence the endpoint is the intended boundary there, every reservoir is nonempty, and allocation selects the certified branch of Lemma~\ref{lem:macro-safe}.
\end{proof}

\begin{lemma}[Run decomposition]\label{lem:run-decomp}
From any macro-entry configuration, the macro-local execution is a schedule and the first-stop truncation of the unique maximal execution, every finite execution is a prefix of that maximal execution, and no final state occurs strictly inside the truncation. In a computation beyond the boot pass, partitioned at its hub visits and its final-state endpoint, each hub-to-hub or hub-to-final interval factors into the schedules of its pay, seek, update, and return macros in order.
\end{lemma}

\begin{proof}
By Lemma~\ref{lem:generated-table}, from any macro-entry configuration there is exactly one maximal execution and every finite execution is a prefix of it. Every table step strictly shortens the configuration, so the maximal execution is finite and ends at a stuck configuration; hence a first stop exists, possibly at time zero, the truncation there is the unique schedule from that configuration, a positive-time visit of the final state is itself a stop, and by the same lemma the final state occurs only as a parking exit targeting $h$. For the factoring, namespaces are disjoint and each macro's exit is the next macro's entry, so each interval splits at its macro-entry visits into the schedules of its pay, seek, update, and return macros in order.
\end{proof}

\subsection{Completeness and Soundness}

\begin{lemma}[Completeness]\label{lem:source-complete}
If $P$ reaches $h$ from its initial configuration, then $V_P$ accepts some input.
\end{lemma}

\begin{proof}
Take the finite halting run. Choose $B$ larger than its number of instructions, choose each $R_i$ larger than the maximum value attained by counter $i$, fix an allocation supplied by Lemma~\ref{lem:fuel-alloc} for the run's full event trace, including the boot pass, and take the physical witness with its reservoir sizes, so that the fence is the genuine padding $\rho$ and there is no suffix.

The witness is the intended boundary at the run's start and its provisioning is adequate, so Lemma~\ref{lem:alloc-real} completes the boot pass, exiting by Lemma~\ref{lem:boot-cert} into the final state when $p_{\mathrm{init}}=h$ or at the intended boundary of the first macro otherwise, and then, applied at every point in turn, chains completed schedules through the intended boundaries of the run into an accepting computation.
\end{proof}

Every computation that has passed the boot pass satisfies, between consecutive visits to the program-control hub and within the certified core:
\begin{enumerate}
  \item[(I1)] after deleting the fuel symbols and omitting the unique state symbol, the data strictly between $\lambda$ and the fence is $1^{b}\,\allowbreak\mathtt L\,1^{m_0}0^{R_0-m_0}\,\mathtt R\,\allowbreak\mathtt L\,1^{m_1}0^{R_1-m_1}\,\mathtt R$ with $0\le b\le B$ and $0\le m_i\le R_i$;
  \item[(I2)] the state belongs to a phase of the single instruction selected at the last hub visit, positioned inside the region that phase may traverse.
\end{enumerate}
Every microstep deletes fuel, moves the state as its phase prescribes, or performs that phase's single rewrite or deletion, so (I1)--(I2) are preserved by induction on steps.

Let $\pi$ be the projection that discards $\lambda$, the fence, and everything beyond it, deletes every fuel symbol, reads off the budget and counter fields, and replaces the hub state by its control, the final state by $h$, and a phase state by the control of the interval's opening hub visit, which (I2) records as selecting its instruction.

\begin{lemma}[Soundness]\label{lem:source-sound}
If $V_P$ accepts an input, then $P$ reaches $h$ from its initial configuration.
\end{lemma}

\begin{proof}
An accepting computation must complete the deterministic boot pass, since the start state lies in the boot namespace, whose only exit is the boot parking step. By Lemma~\ref{lem:boot-cert}, completing the boot pass certifies a physical-witness core up to the fence, so the invariants and the parking property apply within the core, the suffix beyond the fence never contains the state and is inert, and the boot endpoint is the final state when $p_{\mathrm{init}}=h$ or the first certified macro boundary otherwise.

Partition the remaining computation at its hub visits and its final-state endpoint. Invariant (I1) holds at every hub visit, and (I1)--(I2) hold between visits. By Lemma~\ref{lem:run-decomp} each hub-to-hub or hub-to-final interval factors into the schedules of its pay, seek, update, and return macros in order.

We argue by induction along the computation that each boundary so reached is certified and each schedule completed. The first boundary is certified. From a certified boundary the schedule must complete, since an exhaustion stop is a stuck nonfinal configuration, which rejects, against acceptance. By Lemma~\ref{lem:macro-safe} its nonfinal endpoint is a certified macro boundary or a noncertified macro-entry configuration with no applicable command; the latter is stuck nonfinal and rejects likewise, so the endpoint is certified and the induction proceeds to the final parking.

Lemma~\ref{lem:macro-safe}, applied at each boundary in turn, projects each completed schedule under $\pi$, and their composition is exactly one legal instruction of $P$ at the current control. Invariant (I1) keeps each block a legal counter value, and the update macro's exits enforce the intended update and branch. Budget deletion is semantically inert, and termination is forced by the shrinking format.

The projected run starts at the designated initial configuration validated by the boot pass, and the final state arises only as the parking exit of a return targeting $h$, or of boot when the initial control is $h$, so $P$ reaches $h$.
\end{proof}

The construction uses the five data symbols $\{0,1,f,\mathtt L,\mathtt R\}$, radius $k=2$, and the padding $\lambda=\mathtt{LL}$, $\rho=\mathtt{RR}$ independently of $P$. Program controls, macro phases, directions, and boundary tags are stored in states. Taking $\Sigma=D$ permits every data word as a machine input, and an accepted word carries a certified core simulating a genuine halting run.

The start state is the boot entry state, which is nonfinal. Enumerating the finite instruction set and local tag cases constructs the state set and table effectively from $P$ and its initial configuration. Lemmas~\ref{lem:source-complete} and~\ref{lem:source-sound} give $L(V_P)\ne\varnothing$ exactly when $P$ reaches $h$, completing the proof of Lemma~\ref{lem:fuel-compile}.

\section{From Nonemptiness to Weighted Confluence}\label{sec:source}
The binary encoding developed later applies to any finite nonerasing cycle system over a nonempty alphabet with a positive additive weight certificate. This section constructs a uniform family of such systems whose confluence encodes the complement of an arbitrary computably enumerable set. The global correctness argument is completed in Section~\ref{sec:normalization}.

\begin{theorem}[Uniform shrinking verifier]\label{thm:verifier-source}
There are a fixed five-symbol alphabet $D=\Sigma$, a fixed radius $k=2$, fixed padding words $\lambda,\rho\in D^2$, and a total computable map $(e,x)\mapsto M_{e,x}$ to finite deterministic shrinking local machines such that the start state is nonfinal and
\[
  L(M_{e,x})\ne\varnothing
  \quad\Longleftrightarrow\quad
  x\in W_e.
\]
Only the finite state set and local table depend on $(e,x)$.
\end{theorem}

\begin{proof}
By Minsky's construction \cite{Minsky1967}, effectively from $(e,x)$ one obtains a finite deterministic two-counter program $P_{e,x}$ with a designated initial configuration and a halting control $h$ such that $P_{e,x}$ reaches $h$ exactly when $x\in W_e$. If this construction leaves a nonhalting control $p$ without an instruction, complete it effectively with $p:\ C_0:=C_0+1;\ \mathrm{goto}\ p$; a computation formerly stuck at $p$ instead runs forever, so whether $h$ is reached is unchanged. Lemma~\ref{lem:fuel-compile} in Section~\ref{sec:verifier} compiles $P_{e,x}$ into a deterministic shrinking local machine $M_{e,x}=V_{P_{e,x}}$ with the stated fixed alphabet, radius, padding, and nonfinal start state, and with $L(M_{e,x})\ne\varnothing$ exactly when $P_{e,x}$ reaches $h$. Both steps are effective, and only the state set and table depend on $(e,x)$.
\end{proof}

Fix a machine $M=(D,\Sigma,Q,q_0,F,k,\lambda,\rho,\mathcal T)$. Introduce a fresh separator $\#$, a raw copy $\widehat\Sigma=\{\widehat a:a\in\Sigma\}$ with $\Dbar=D\mathbin{\dot\cup}\widehat\Sigma$, and two copied state sets $Q^0,Q^1$. The controls are $\Ctrl=\{X,d_0,d_1,A_0,A_1,N,E\}\mathbin{\dot\cup}Q^0\mathbin{\dot\cup}Q^1$, and the declared alphabet is $\Gamma_M=\{\#\}\mathbin{\dot\cup}\Dbar\mathbin{\dot\cup}\Ctrl$. For $w=a_1\cdots a_n\in\Sigma^*$ put $\widehat w=\widehat{a_1}\cdots\widehat{a_n}$. The intended witness is $\cyc{\#\widehat wX}$. The factor $X\#$ across its displayed end supports the only deliberate branch.

All ranges in Table~\ref{tab:wrapper-rules} are finite. In particular, the boundary rules \textup{(R2)} and \textup{(R3)} are added only for nonfinal copied states; final states use \textup{(A)} and have no competing boundary rejection.

\begin{table}[!ht]
\caption{Rule families of the wrapper $R_M$.}
\label{tab:wrapper-rules}
\centering
\small
\setlength{\tabcolsep}{3.5pt}
\renewcommand{\arraystretch}{1.08}
\begin{tabularx}{\textwidth}{@{}l >{\raggedright\arraybackslash}p{0.31\textwidth} >{\raggedright\arraybackslash}X@{}}
\hline
Name & Rule & Range and condition \\
\hline
$(\mathrm L_b)$ & $X\#\to d_b\rho\#$ & $b\in\{0,1\}$ \\
$(\mathrm{LX})$ & $Xd\to N$ & $d\in\Dbar$ \\
$(\mathrm{D1})$ & $\widehat a\,d_b\to d_b\,a$ & $a\in\Sigma$, $b\in\{0,1\}$ \\
$(\mathrm{D2})$ & $\#d_b\to\#\lambda q_0^b$ & $b\in\{0,1\}$ \\
$(\mathrm{D3})$ & $dd_b\to N$ & $d\in D$, $b\in\{0,1\}$ \\
$(\mathrm M)$ & $\alpha q^b\beta\to u\,(q')^b\,v$ & $\alpha q\beta\to uq'v\in\mathcal T$, $b\in\{0,1\}$ \\
$(\mathrm A)$ & $q^b\to A_b$ & $q\in F$, $b\in\{0,1\}$ \\
$(\mathrm{R1})$ & $\alpha q^b\beta\to N$ & $q\in Q\setminus F$, $b\in\{0,1\}$, $\alpha,\beta\in\Dbar^k$, and no table entry has left side $\alpha q\beta$ \\
$(\mathrm{R2})$--$(\mathrm{R3})$ & $\#\,u\,q^b\to\#N$, $q^b\,v\,\#\to N\#$ & $q\in Q\setminus F$, $b\in\{0,1\}$, $u,v\in\Dbar^{<k}$ \\
$(\mathrm S)$ & $dS\to S$, $Sd\to S$ & $S\in\{A_0,A_1,N,E\}$, $d\in\Dbar$ \\
$(\mathrm C)$ & $czc'\to E$ & $c,c'\in\Ctrl$, $z\in\Dbar^{<k+1}$ \\
$(\mathrm{CE})$ & $cE\to E$, $Ec\to E$ & $c\in\Ctrl$ \\
\hline
\end{tabularx}
\end{table}

No rule creates or deletes a separator, no rule transfers a symbol across a separator, and every right-hand side is nonempty. A cycle containing a separator decomposes as $\#s_1\#s_2\cdots\#s_j$ with the $s_i$ its \emph{separator-delimited segments}; in segment notation such as $\#sX\#$ the displayed separators are consecutive occurrences on the cycle, coinciding when the cycle has exactly one separator. Since all ranges are finite, $R_M$ is finite and computable from $M$.

On the intended witness the two branches unfold, for $b\in\{0,1\}$, as
\[
  \cyc{\#\widehat wX}
  \cstep
  \cyc{\#\widehat w\,d_b\rho}
  \cstep^{*}
  \cyc{\#d_b w\rho}
  \cstep
  \cyc{\#\lambda q_0^{b} w\rho}
  \cstep^{*}
  \begin{cases}
    \cyc{\#A_b} & \text{if } w\in L(M),\\[2pt]
    \cyc{\#N}   & \text{if } w\notin L(M).
  \end{cases}
\]
The launch chooses a branch, the decoder converts the raw input, the two state copies simulate $M$ deterministically, and only the acceptance outcome remembers the branch. Lemma~\ref{lem:W-launch} below makes this exact. Its rejection analysis also covers malformed launch segments.

\begin{lemma}[Exact launch behaviour]\label{lem:W-launch}
For a one-control segment $\#sX\#$ and a chosen branch $b$:
\begin{enumerate}
  \item if $s\notin\widehat\Sigma^*$, every maximal derivation in the branch ends in $\#N\#$;
  \item if $s=\widehat w$ with $w\in\Sigma^*$, the branch reaches $\#\lambda q_0^{b} w\rho\#$, simulates the unique computation of $M$ on $w$, and every maximal derivation then ends in $\#A_b\#$ if $w\in L(M)$ and in $\#N\#$ otherwise.
\end{enumerate}
\end{lemma}

\begin{proof}
Launch gives $\#s d_b\rho\#$. Rule \textup{(D1)} moves $d_b$ left through the maximal raw suffix of $s$, converting it into genuine input data. If $s$ is not entirely raw, the next predecessor is genuine data and \textup{(D3)} creates $N$. If $s=\widehat w$, decoding ends with \textup{(D2)} and reaches the segment $\#\lambda q_0^{b}w\rho\#$. Thereafter the wrapper follows the source computation configuration by configuration in branch $b$. For a copied nonfinal state with a full $k$-by-$k$ data window, a listed source window has its unique machine rule and no default, while an unlisted window has its \textup{(R1)} default and no machine rule. If no full window exists because the nonfinal state lies within $k-1$ data symbols of a separator, \textup{(R2)} or \textup{(R3)} rejects exactly that source configuration as stuck; when both apply, each result sweeps to $\#N\#$. A copied final state instead enables only \textup{(A)}, independently of its distance from a separator. Thus wrapper rejection agrees exactly with source-machine rejection. Once $N$ or $A_b$ is created in the one-control segment, only the sweep rules \textup{(S)} apply; each step deletes one data symbol, any left and right sweep peak commutes in one step, and termination yields the unique normal form $\#N\#$ or $\#A_b\#$.
\end{proof}

\begin{lemma}[Explicit weight decrease]\label{lem:wrapper-weight}
The wrapper $R_M$ has an effectively computable positive additive weight certificate.
\end{lemma}

\begin{proof}
Give every symbol of $D$ and every copied state weight $2$, every raw input symbol weight $3$, and $\omega(\#)=\omega(A_0)=\omega(A_1)=\omega(N)=\omega(E)=1$. Let $W=2|\lambda|+3$, $\omega(d_0)=\omega(d_1)=W$, and $\omega(X)=W+2|\rho|+1$.

A decoder move, decoder completion, and launch each lower weight by one. A genuine machine left side has weight $4k+2$, while its right side has weight at most $4k$. Every malformed-launch, malformed-decoder, full-window rejection, boundary rejection, acceptance, sweep, collision, and absorption rule strictly decreases weight by direct comparison: it deletes a symbol of positive weight, replaces a symbol by a strictly lighter one, or replaces at least two positive-weight controls by the unit-weight error control. Hence every rule decreases weight strictly.
\end{proof}

\subsection{Global Normalization}\label{sec:normalization}
The reduction must control every cycle over $\Gamma_M$, not merely an intended launch cycle. Termination follows from the weight certificate (Lemma~\ref{lem:wrapper-weight}); the lemmas below show that the error control is absorbing, classify components by control count, treat separator-free cycles directly, and assemble separator-delimited segments by asynchronous product, with the launch fork's exact behaviour given by Lemma~\ref{lem:W-launch} above.

If a cycle contains at least one separator, choose a representative of the form $\#s_1\#s_2\cdots\#s_j$. A \emph{control occurrence} is an occurrence of a letter in $\Ctrl$. Until a rule of type \textup{(C)} or \textup{(CE)} is used, every rule preserves the number of controls in its segment.

\begin{lemma}[The error control persists]\label{lem:W-E}
Once a separator-delimited segment contains $E$, every descendant contains $E$, and its unique normal form is the one-letter segment $E$. Once a separator-free cycle contains $E$, every descendant contains $E$, and its unique cyclic normal form is $\cyc{E}$.
\end{lemma}

\begin{proof}
No right-hand side lacking $E$ has a left-hand side containing $E$, so $E$ cannot disappear. In either kind of component, if another symbol remains, a nearest neighbour of an occurrence of $E$ is either data, deleted by a sweep, or a control, deleted by a collision or control-absorption rule. Hence a component containing $E$ is reducible unless it is exactly $E$. Termination gives the claim.
\end{proof}

\begin{lemma}[Every multi-control component collapses]\label{lem:W-collision}
Every separator-delimited segment with at least two controls has the unique normal form $E$. Every separator-free cycle with at least two controls has the unique cyclic normal form $\cyc{E}$.
\end{lemma}

\begin{proof}
If the component already contains $E$, Lemma~\ref{lem:W-E} applies. Assume first that a separator-delimited segment contains no $E$, and consider a maximal reduction from it, which ends in a normal form by termination. Until the first use of \textup{(C)}, rule \textup{(CE)} is unavailable and every rule preserves the number of controls. Suppose for contradiction that \textup{(C)} is never used. In the resulting normal form, if two consecutive controls had fewer than $k+1$ data symbols between them, \textup{(C)} would apply. Hence every consecutive control gap would contain at least $k+1$ data symbols.

Inspect any control in that alleged normal form. An $X$ is followed by data, enabling \textup{(LX)}, or by the right separator, enabling a launch. A decoder $d_b$ has immediately to its left a raw letter, genuine data, or the left separator, enabling respectively \textup{(D1)}, \textup{(D3)}, or \textup{(D2)}. A copied final state enables \textup{(A)}. A nonfinal copied state either has $k$ data symbols on both sides, in which case exactly one of the simulation rule and its complementary \textup{(R1)} applies, or lies within $k-1$ data symbols of a separator, enabling \textup{(R2)} or \textup{(R3)}. Finally, each of $A_0,A_1,N$ can sweep towards another control or is already within collision range.

Thus the terminal object is not a normal form, a contradiction. Every maximal reduction therefore uses \textup{(C)}, creates $E$, and then has the unique normal form $E$ by Lemma~\ref{lem:W-E}.

For a separator-free cycle, dispose first of the case in which $E$ is already present. Otherwise consider the controls, consecutive in cyclic order. If a cyclic gap is shorter than $k+1$, \textup{(C)} applies; if every gap has length at least $k+1$, the same case analysis enables a rule at each control. Hence no maximal reduction can avoid \textup{(C)}. The no-winding convention ensures that its left side uses two distinct control occurrences rather than reusing one occurrence after a full turn, and Lemma~\ref{lem:W-E} again gives the unique cyclic normal form $\cyc{E}$.
\end{proof}

\begin{lemma}[One-control determinism apart from launch]\label{lem:W-one}
A separator-delimited segment with no controls is irreducible. A segment with exactly one control has a unique normal form, except possibly when that control is $X$ immediately before the right separator. In that case there are exactly the two launch choices, and after a branch is chosen all remaining nondeterminism consists of commuting sweeps or the two boundary rejections, which have the same normal form $N$.
\end{lemma}

\begin{proof}
No left side consists only of data. For one control, its type and bounded neighbourhood determine the next rule. A decoder is governed by its immediate predecessor. A final copied state accepts. A nonfinal copied state with a full $k$-by-$k$ window has exactly one of a genuine table rule and the complementary rejection; near a separator it has the appropriate boundary rejection.

If both boundary rules apply to a short segment $\#u q^{b}v\#$ with $|u|,|v|<k$, their targets $\#Nv\#$ and $\#uN\#$ both sweep to $\#N\#$. Each of $A_0,A_1,N,E$ deletes adjacent data in either order but leaves the same single control. An $X$ followed by data has only the malformed-launch rule, whereas $X\#$ has exactly the two launch rules. Collision rules require two controls.
\end{proof}

\begin{lemma}[Asynchronous product of segments]\label{lem:W-product}
For a cycle containing at least one separator, rewriting is the asynchronous product, modulo cyclic permutation, of the relations on its separator-delimited segments. If every reachable segment has a unique segment normal form, then the whole cycle has a unique cyclic normal form.
\end{lemma}

\begin{proof}
A rule without $\#$ is internal to one segment. Every rule containing $\#$ has it at one end of both sides and changes symbols on only one side. Steps in nonadjacent segments are disjoint; steps in adjacent segments may share the unchanged separator but commute. Reductions are therefore precisely interleavings of component reductions. The tuple of segment normal forms is unique up to cyclic choice of the first separator, which is exactly rotation.
\end{proof}

\begin{lemma}[Separator-free cycles]\label{lem:W-no-sep}
Every separator-free cycle has a unique normal form.
\end{lemma}

\begin{proof}
With no controls, no rule applies. With at least two controls, Lemma~\ref{lem:W-collision} gives $\cyc{E}$.

Suppose there is exactly one control. The one-letter cycles $\cyc{X}$ and $\cyc{d_b}$ are irreducible. Otherwise an $X$ followed by data rejects to $N$. A decoder repeatedly converts raw predecessors, each \textup{(D1)} step replacing the raw letter to its left by a genuine one, so the number of raw letters strictly decreases. Unless the cycle is the one-letter case already handled, the decoder meets genuine data after at most one circuit and \textup{(D3)} rejects.

A copied final state accepts. A nonfinal copied state has exactly one simulation/default rule when a simple arc containing $k$ distinct data positions on each side exists; if the cycle is too short, no such left side can match. Finally, $A_b,N,E$ sweep all data. Each case has one normal form.
\end{proof}

\begin{theorem}[Weighted confluence reduction]\label{thm:weighted-fork}
From $M$ one effectively constructs a finite nonerasing additively weight-decreasing cycle system $R_M$ such that
\[
  R_M\text{ is confluent}
  \quad\Longleftrightarrow\quad
  L(M)=\varnothing.
\]
Consequently there is a total computable family $(S_{e,x},\omega_{e,x})$ satisfying
\[
  S_{e,x}\text{ is confluent}
  \quad\Longleftrightarrow\quad
  x\notin W_e.
\]
\end{theorem}

\begin{proof}
The relation terminates by Lemma~\ref{lem:wrapper-weight} and Proposition~\ref{prop:certificate}.

If $w\in L(M)$, the two launch branches reach the segment normal forms $\#A_0\#$ and $\#A_1\#$ by Lemma~\ref{lem:W-launch}. Segment derivations lift to the one-separator cycle $\cyc{\#\widehat wX}$ by Lemma~\ref{lem:W-product}, and each $\cyc{\#A_b}$ is irreducible, since sweeps need an adjacent data symbol and collisions need two control occurrences. The cycle thus has two distinct normal forms, so $R_M$ is not confluent by Lemma~\ref{lem:unique-nf}.

Assume $L(M)=\varnothing$. Section~\ref{sec:normalization} proves that every separator-delimited segment other than a launch segment has a unique normal form and that separator-free cycles do as well. In the sole launch case, both branches reject by Lemma~\ref{lem:W-launch} because the language is empty. Lemma~\ref{lem:W-product} transfers the segment normal forms to every cycle with separators, so every cycle has a unique normal form and the relation is confluent by Lemma~\ref{lem:unique-nf}. Composition with Theorem~\ref{thm:verifier-source} gives the uniform family.
\end{proof}

\section{A Confluence-Preserving Binary Compiler}\label{sec:compiler}

\subsection{A Binary Encoding}\label{sec:binary-code}
This section translates a finite nonerasing weight-decreasing cycle system into a binary system whose rules strictly reduce ordinary length. The encoding uses the lengths of zero runs to represent source letters. We first define the code and prove length decrease, then show that well-formed binary cycles simulate the source system exactly. Example~\ref{ex:binary} at the end of Section~\ref{sec:transfer} instantiates the construction.

Let $A$ be a finite nonempty alphabet, let $R\subseteq A^+\times A^+$ be finite and nonerasing, and let $\omega:A\to\N_{>0}$ be a positive additive certificate with $\omega(\ell)>\omega(r)$ for every $\ell\to r$. We construct a binary system $\mathsf B(R,\omega)$.

Enumerate $A=\{a_0,\ldots,a_{m-1}\}$. If $w=a_{j_1}\cdots a_{j_t}$, put $\iota(w)=j_1+\cdots+j_t$, with $\iota(\eps)=0$. Choose an integer $K$ with $K>m$ and, for every rule $\ell\to r$,
\[
  K>\left|\iota(\ell)-\iota(r)+3(|\ell|-|r|)\right|.
\]
Define code lengths $c(a_i)=K\omega(a_i)+i$. They are pairwise distinct: equal weights leave distinct residues $i$, while unequal weights differ by at least $K-(m-1)>0$. They also satisfy $c(a_i)\ge2$. Put $C=\{c(a):a\in A\}$ and $T=1+\max C$, and define the binary morphism $\enc(a)=1\,0^{c(a)}\,1\,0$. Thus a codeword contributes alternately a long zero gap of length $c(a)$ and a unit zero gap. The \emph{simulation rules} are $\enc(\ell)\to\enc(r)$ for $\ell\to r\in R$. They are nonerasing because $R$ is nonerasing.

\begin{lemma}[Bit-length decrease]\label{lem:bitlength}
For every source rule $\ell\to r$, $|\enc(\ell)|>|\enc(r)|$. Hence the simulation rules are strictly length reducing.
\end{lemma}

\begin{proof}
For every word $w$, $|\enc(w)|=K\omega(w)+\iota(w)+3|w|$, so
\[
  |\enc(\ell)|-|\enc(r)|
  =K\bigl(\omega(\ell)-\omega(r)\bigr)
  +\iota(\ell)-\iota(r)+3(|\ell|-|r|).
\]
The weight difference is at least one, and the choice of $K$ makes the right-hand side positive.
\end{proof}

\subsubsection{Gap Sequences and Valid Cycles}

For a binary cycle $C_0$ containing at least one $1$, list those occurrences in cyclic order. The corresponding \emph{gap sequence} $g(C_0)=(g_1,\ldots,g_q)$ records the number of zeros from each $1$ to the next. It is defined up to cyclic shift.

\begin{definition}[Valid encoded cycle]\label{def:valid}
The empty cycle is \emph{valid}, a nonempty cycle containing no $1$ is invalid, and a cycle containing at least one $1$ is valid exactly when its cyclic gap sequence alternates $c(a_1),1,\ldots,c(a_s),1$ for some $s\ge1$ and letters $a_1,\ldots,a_s\in A$.
\end{definition}

Since every member of $C$ is at least two, the long and unit gaps determine the phase uniquely. Since the values $c(a)$ are pairwise distinct, every long gap determines its source letter.

\begin{lemma}[Cyclic injectivity]\label{lem:cyclic-injective}
For $u,v\in A^*$, $\cyc{\enc(u)}=\cyc{\enc(v)}$ if and only if $\cyc{u}=\cyc{v}$. Moreover, every nonempty valid binary cycle has a unique decoding to a source cycle.
\end{lemma}

\begin{proof}
Since letter codes are nonempty, $\cyc{\enc(u)}$ is empty exactly when $u$ is, and the empty cases match; assume $u,v$ nonempty. The alternating gap sequence of $\enc(u)$ is obtained from the cyclic letter sequence of $u$ by replacing each $a$ by the pair $(c(a),1)$. Long gaps cannot be confused with unit gaps, and $c$ is injective. Thus cyclic shifts are the only ambiguity.
\end{proof}

\begin{lemma}[Exact simulation]\label{lem:exact-simulation}
On valid cycles, the binary simulation relation is graph-isomorphic to the source cycle relation under $\cyc{u}\mapsto\cyc{\enc(u)}$.
\end{lemma}

\begin{proof}
Let $C=\cyc{\enc(u)}$ be valid. A simulation left-hand side begins with a $1$ followed by a long gap, so its match in $C$ can start only at the first $1$ of a codeword. The alternating gaps then identify $|\ell|$ consecutive source letters, which must spell $\ell$. Because the matching arc is simple, it cannot revisit its initial codeword, which is exactly the source no-winding condition $|\ell|\le|u|$. Hence, after rotating $u$, there is a word $t$ with $u=\ell t$, and the binary step is $\cyc{\enc(\ell t)}\to\cyc{\enc(rt)}$.

Conversely, every source step from $\cyc{u}$ with $u=\ell t$ yields exactly this binary step, and since $|\enc(\ell)|\le|\enc(\ell t)|$ its match is a simple arc. Cyclic injectivity (Lemma~\ref{lem:cyclic-injective}) therefore makes $\cyc{u}\mapsto\cyc{\enc(u)}$ an isomorphism on both vertices and reduction edges.
\end{proof}

\subsection{Malformed Cycles and Confluence Preservation}\label{sec:transfer}
Exact simulation on well-formed encodings is not sufficient, because confluence is required on every binary cycle: an injective fixed-length morphism would synchronize the intended words but say nothing about malformed cycles, whose unused bit patterns can act as inert separators and expose ordinary nonconfluence inside arbitrary segments. This section therefore adds finitely many cleanup rules. They detect every reducible malformed gap pattern and create the error pair $11$, from which a unique normal form is reached. This step upgrades an embedding of the intended computation to an equivalence of global confluence properties.

Let $\Eset$ consist of the following rules:
\begin{align*}
  0^T &\longrightarrow 11, \tag{E1}\\
  1 0^n 1 &\longrightarrow 11
     &&(2\le n<T,\ n\notin C), \tag{E2}\\
  1 0^p 1 0^q &\longrightarrow 11
     &&(p,q\in C), \tag{E3}\\
  10101 &\longrightarrow 11, \tag{E4}\\
  110 &\longrightarrow 11, \qquad
  111 \longrightarrow 11. \tag{E5}
\end{align*}
Rule \textup{(E3)} deliberately has no final $1$: a rule ending in a third delimiter could never match a cycle with exactly two $1$'s and two code-sized gaps, because a left-hand side may not wind around and reuse its initial $1$. The suffix-free form uses every position of such a cycle exactly once and closes this boundary case.

Every rule in $\Eset$ strictly decreases length: $T\ge3$, the left side of \textup{(E2)} has length at least four, the left side of \textup{(E3)} has length at least six, and the remaining claims are immediate.

Define the full binary compiler by
\[
  \mathsf B(R,\omega)
  =\{\enc(\ell)\to\enc(r):\ell\to r\in R\}
    \ \cup\ \Eset.
\]

\begin{lemma}[Valid cycles avoid cleanup rules]\label{lem:valid-closed}
No rule of $\Eset$ applies to a valid cycle. A simulation step from a valid cycle has a valid target. Hence valid cycles form a closed subgraph of $\mathsf B(R,\omega)$.
\end{lemma}

\begin{proof}
A valid cycle has zero gaps only in $C\cup\{1\}$, never has two code gaps consecutively, never has two unit gaps consecutively, and has no adjacent $1$'s. Its longest zero run is less than $T$. Thus none of \textup{(E1)}--\textup{(E5)} matches. Closure under simulation was proved in Lemma~\ref{lem:exact-simulation}.
\end{proof}

The next lemma is the key reverse synchronization statement.

\begin{lemma}[Validity cannot be repaired by simulation]\label{lem:backward-valid}
If a simulation step has a valid target, then its source was valid. Equivalently, simulation preserves invalidity.
\end{lemma}

\begin{proof}
Choose representatives for which the step is $\enc(\ell)z\to\enc(r)z$. If $z=\eps$, the source is $\cyc{\enc(\ell)}$ and is valid.

Suppose $z\ne\eps$ and the target is valid. Since $r$ is nonempty, the first $1$ of $\enc(r)$ is followed by a long gap and is therefore a codeword start in the unique valid parsing of the target. The parsing follows all of $\enc(r)$. After its final unit spacer, the first symbol of $z$ must begin the next codeword, and continuing around the target shows that $z=\enc(t)$ for some $t\in A^*$. Hence the source is $\cyc{\enc(\ell)\enc(t)}=\cyc{\enc(\ell t)}$ and is valid. The nonerasing hypothesis is essential: an empty inserted word would provide no internal long gap from which to recover the phase.
\end{proof}

\subsubsection{Classification of Malformed Cycles}

Let $\Iset$ be the following finite family of isolated malformed cycles:
\[
  \Iset=
  \{\cyc{0^n}:1\le n<T\}
  \ \cup\
  \{\cyc{10^n}:0\le n<T\}
  \ \cup\
  \{\cyc{1010}\}.
\]
The middle family consists of cycles with exactly one occurrence of $1$. No simulation left-hand side can occur in a member of $\Iset$: an encoded nonempty left-hand side contains at least two occurrences of $1$ and at least one long gap.

\begin{lemma}[Classification of malformed cycles]\label{lem:garbage-class}
Every binary cycle belongs to exactly one of the following four classes:
\begin{enumerate}
  \item a valid encoded cycle;
  \item the error normal form $\cyc{11}$;
  \item an isolated malformed cycle in $\Iset$;
  \item a cycle admitting a rule from $\Eset$.
\end{enumerate}
\end{lemma}

\begin{proof}
The empty cycle is valid. A nonempty all-zero cycle either has length less than $T$, in which case it lies in $\Iset$, or admits \textup{(E1)}. Now suppose the cycle contains $q\ge1$ occurrences of $1$, with cyclic zero gaps $g_1,\ldots,g_q$. If some gap is at least $T$, rule \textup{(E1)} applies. Assume all gaps are less than $T$. If $q=1$ and its sole gap is zero, the cycle is $\cyc{1}\in\Iset$. Otherwise, if some gap is zero, two distinct consecutive occurrences of $1$ are adjacent. The length-two case is $\cyc{11}$. Every longer such cycle has either a factor $110$ or a factor $111$, so \textup{(E5)} applies.

Assume henceforth that every gap is positive. If $q=1$, the cycle is $\cyc{10^n}$ with $n<T$ and belongs to $\Iset$. If some gap is neither one nor a member of $C$, then it lies between two and $T-1$ and rule \textup{(E2)} applies. Thus it remains to consider cyclic sequences over the two gap types ``unit'' and ``code''. Two consecutive code gaps trigger \textup{(E3)}. If $q\ge3$, two consecutive unit gaps trigger \textup{(E4)}. Therefore a cycle with $q\ge3$ and no cleanup redex must alternate code and unit gaps, which is exactly validity. When $q=2$, two code gaps trigger \textup{(E3)}, one gap of each type gives a valid one-letter encoding, and two unit gaps give the exceptional cycle $\cyc{1010}$. The four classes are disjoint by construction.
\end{proof}

\begin{lemma}[The error pair persists]\label{lem:error-persists}
If a cycle contains adjacent $1$'s, then every one-step descendant contains adjacent $1$'s. Every such cycle has the unique normal form $\cyc{11}$.
\end{lemma}

\begin{proof}
Every cleanup rule creates or preserves the factor $11$. A simulation left-hand side contains no adjacent $1$'s and ends in $0$. Hence a simulation match cannot contain both members of an existing pair or end at its first member. The only possible overlap starts at the second $1$. The nonempty encoded right-hand side starts with $1$, leaving it adjacent to the first, untouched $1$.

If a cycle containing $11$ has length greater than two, the symbol following an adjacent pair is either $0$ or $1$, enabling $110\to11$ or $111\to11$. Repeated use strictly shortens the cycle until $\cyc{11}$ remains. Persistence shows that no other normal form is reachable.
\end{proof}

\begin{lemma}[Malformed reducible cycles normalize to error]\label{lem:invalid-nf}
Every invalid cycle outside $\Iset$ has the unique normal form $\cyc{11}$ in $\mathsf B(R,\omega)$.
\end{lemma}

\begin{proof}
The full binary system is strictly length reducing, so every reduction terminates. By Lemma~\ref{lem:backward-valid}, a simulation path from an invalid source remains invalid. A simulation target contains the nonempty block $\enc(r)$, hence at least two $1$'s and a long code gap. It therefore cannot be a member of $\Iset$, and it cannot be $\cyc{11}$. If a maximal reduction has not yet created $11$, its current invalid cycle, which is not one of the isolated exceptions, admits a cleanup rule by Lemma~\ref{lem:garbage-class}. Once that rule is used, Lemma~\ref{lem:error-persists} gives the unique normal form $\cyc{11}$.
\end{proof}

\begin{theorem}[Confluence-preserving binary encoding]\label{thm:binary-transfer}
There is an effective transformation that maps every finite nonerasing additively weight-decreasing cycle system $(A,R,\omega)$ over a nonempty alphabet to a finite system
\[
  \mathsf B(R,\omega)\subseteq\B^+\times\B^+
\]
such that:
\begin{enumerate}
  \item every binary rule strictly decreases ordinary word length;
  \item the valid binary subgraph is isomorphic to the source cycle-reduction graph;
  \item $\mathsf B(R,\omega)$ is confluent if and only if $R$ is confluent.
\end{enumerate}
\end{theorem}

\begin{proof}
Effectivity and strict length decrease follow from the explicit finite construction and Lemma~\ref{lem:bitlength}. Exact simulation and closure of valid cycles follow from Lemmas~\ref{lem:exact-simulation} and~\ref{lem:valid-closed}.

Suppose first that $\mathsf B(R,\omega)$ is confluent. Any source peak maps to a binary peak inside the closed valid subgraph. A binary join of its branches remains valid and decodes, by Lemmas~\ref{lem:cyclic-injective} and~\ref{lem:exact-simulation}, to a source join. Hence $R$ is confluent.

Conversely, suppose $R$ is confluent. A valid binary source has only simulation descendants, and exact simulation gives confluence there. A source in $\Iset$ is irreducible. The error cycle $\cyc{11}$ is irreducible. Every other invalid source has the unique normal form $\cyc{11}$ by Lemma~\ref{lem:invalid-nf}. Thus every binary cycle has a unique normal form. Since the system terminates, it is confluent by Lemma~\ref{lem:unique-nf}.
\end{proof}

\begin{example}\label{ex:binary}
Let $A=\{a_0,a_1\}$ with $a_0=a$ and $a_1=b$, take unit weights, and let $R=\{ab\to a,\ ab\to b\}$. Both rules decrease weight. The constant $K$ must exceed $m=2$ and, for the two rules, $|\iota(ab)-\iota(a)+3|=4$ and $|\iota(ab)-\iota(b)+3|=3$. Take $K=5$. Then $c(a)=5$, $c(b)=6$, $C=\{5,6\}$, and $T=7$, so $\enc(a)=10^{5}10$ and $\enc(b)=10^{6}10$, and the two simulation rules $\enc(ab)\to\enc(a)$ and $\enc(ab)\to\enc(b)$ shorten $17$ bits to $8$ and to $9$. The cleanup system consists of $0^{7}\to11$, the three rules $10^{n}1\to11$ for $n\in\{2,3,4\}$, the four rules $10^{p}10^{q}\to11$ for $p,q\in\{5,6\}$, and the fixed rules \textup{(E4)} and \textup{(E5)}.

The valid cycle $\cyc{\enc(ab)}$, with gap sequence $(5,1,6,1)$, reproduces the source fork exactly: its two descendants $\cyc{\enc(a)}$ and $\cyc{\enc(b)}$ are distinct irreducibles, so nonconfluence transfers. For example, the malformed cycle $\cyc{10010}$, with gap sequence $(2,1)$, is caught by rule \textup{(E2)} and drains to the error normal form:
\[
  \cyc{10010}\cstep\cyc{110}\cstep\cyc{11}.
\]
\end{example}

\section{Classification}\label{sec:classification}
This section proves the main undecidability result. We first establish the arithmetical upper bound, then combine the weighted construction with the binary encoding. Afterwards we extend the result to every larger fixed alphabet and briefly record the decidable unary boundary.

\subsection{Binary Systems}
For each fixed finite alphabet $\Omega$, fix a total coding of finite rule lists over $\Omega$. Let $\mathcal C_\Omega$ contain exactly the codes of finite systems $R\subseteq\Omega^+\times\Omega^+$ that are strictly length reducing and whose induced cycle relation is confluent. Invalid codes and codes containing a nondecreasing rule are outside $\mathcal C_\Omega$.

\begin{lemma}[Arithmetical upper bound]\label{lem:pi-upper}
For every fixed finite alphabet $\Omega$, the complement of $\mathcal C_\Omega$ is computably enumerable. Hence $\mathcal C_\Omega\in\Pi^0_1$.
\end{lemma}

\begin{proof}
Syntax and strict length decrease are decidable. For a valid system, enumerate cyclic words and their one-step peaks. From a source of length $n$, every step lowers length, so the complete descendant graph is finite and can be computed (Proposition~\ref{prop:certificate}). Nonjoinability is therefore decidable for each enumerated peak. By Newman's lemma (Theorem~\ref{thm:newman}), a nonconfluent length-reducing system has a nonjoinable local peak, and the enumeration eventually finds it.
\end{proof}

\begin{theorem}[Undecidability over two letters]\label{thm:pi-complete}
For the fixed alphabet $\B$, confluence of finite nonerasing strictly length-reducing cycle-rewriting systems is $\Pi^0_1$-complete, and therefore undecidable. Its complement is $\Sigma^0_1$-complete.
\end{theorem}

\begin{proof}
Let $A\in\Pi^0_1$ and choose $e$ with $x\in A$ if and only if $x\notin W_e$. The weighted confluence reduction produces $(S_{e,x},\omega_{e,x})$ such that $S_{e,x}$ is confluent if and only if $x\notin W_e$. The compiled system $\mathsf B(S_{e,x},\omega_{e,x})$ lies in $\B^+\times\B^+$ and is strictly length reducing by Theorem~\ref{thm:binary-transfer}, so its code satisfies the syntactic conditions of $\mathcal C_{\B}$, and Theorem~\ref{thm:binary-transfer} gives
\[
  x\in A
  \quad\Longleftrightarrow\quad
  \mathsf B(S_{e,x},\omega_{e,x})\in\mathcal C_{\B}.
\]
The map is total and computable. Lemma~\ref{lem:pi-upper} gives membership in $\Pi^0_1$, and complementation gives the $\Sigma^0_1$ statement.
\end{proof}

\subsection{Larger Fixed Alphabets}

\begin{lemma}[Binary right-hand sides]\label{lem:rhs-two}
Every right-hand side produced by the binary encoding has length at least two.
\end{lemma}

\begin{proof}
A simulation right-hand side is $\enc(r)$ with $r\ne\eps$, so it contains a complete codeword. Every cleanup rule has right-hand side $11$.
\end{proof}

\begin{corollary}[All fixed alphabets]\label{cor:all-alphabets}
For every fixed finite alphabet $\Omega$ with $|\Omega|\ge2$, confluence of finite nonerasing strictly length-reducing cycle-rewriting systems over $\Omega$ is $\Pi^0_1$-complete, and therefore undecidable.
\end{corollary}

\begin{proof}
Choose distinct letters of $\Omega$ and name them $0$ and $1$. Given a compiled binary system $R$, let $Y=\Omega\setminus\B$ be the set of \emph{extra} letters, and extend $R$ by the rules
\begin{equation*}
  xcc'\longrightarrow11
  \qquad(x\in Y,\ c,c'\in\Omega). \tag{Ab}
\end{equation*}
They are nonerasing and have shape $3\to2$. No right-hand side contains an extra letter, and no old left-hand side contains one. Pure binary cycles therefore form a closed subgraph identical to the original binary system.

Consider a cycle containing an extra letter. Every left-hand side of the extended system has length at least three, since simulation left sides contain a full codeword, cleanup left sides have length at least three, and absorbers have length exactly three. Hence every cycle of length at most two is irreducible. If the length is at least three, the three-position arc beginning at any extra occurrence enables \textup{(Ab)}. An old binary step avoids every extra occurrence, and if an extra remains afterwards, the target still has length at least three: it contains the untouched extra and a right-hand side of length at least two, by Lemma~\ref{lem:rhs-two} for the binary rules and by construction for the absorbers. Thus no reduction from an extra-containing cycle of length at least three reaches a cycle of length at most two while an extra remains.

The extended system is strictly length reducing, so every maximal reduction ends in a normal form. By the previous paragraph, a cycle in such a reduction that still contains an extra has length at least three and is therefore reducible. Hence the terminal normal form is extra-free, and the reduction contains a step removing its last extra. Binary rules leave extras untouched, so that step is an absorber step, and its target is a pure binary cycle containing the factor $11$, whose unique normal form is $\cyc{11}$ by Lemma~\ref{lem:error-persists} applied inside the closed binary subgraph. Every maximal reduction from an extra-containing cycle of length at least three therefore ends at $\cyc{11}$, its unique normal form.

If the extended system is confluent, so is the binary subsystem: descendants of binary cycles are binary, so every join supplied by the extended relation is a binary join. Conversely, if the binary subsystem is confluent, then every binary cycle has a unique normal form by Lemma~\ref{lem:unique-nf}, every extra-containing cycle of length at most two is its own normal form, and every longer extra-containing cycle has the unique normal form $\cyc{11}$. Since the extended system terminates, it is confluent by Lemma~\ref{lem:unique-nf}. Composing the reduction in the proof of Theorem~\ref{thm:pi-complete} with this extension establishes hardness, and Lemma~\ref{lem:pi-upper} gives the upper bound.
\end{proof}

\subsection{The Unary Boundary}\label{sec:unary}
Let the alphabet be $\{a\}$, and write the rules as $a^{p_i}\to a^{q_i}$ with $p_i\ge1$ and $q_i\ge0$. Put $d_i=p_i-q_i$. A unary cycle is determined by its length $n$, and rule $i$ sends $n$ to $n-d_i$ whenever $n\ge p_i$.

\begin{theorem}[Decidability over one letter]\label{thm:unary}
Termination is decidable for finite unary cycle-rewriting systems, and confluence is decidable on the terminating subclass. In particular, confluence is decidable for finite strictly length-reducing unary systems.
\end{theorem}

\begin{proof}
The system terminates exactly when every $d_i$ is positive: a zero difference gives a self-loop, a negative difference can be repeated forever, and when every $d_i$ is positive each step strictly lowers length, so no infinite reduction exists. This decides termination. Assume from now on that the system terminates. For two rules $i,j$, put $B_{ij}=\max\{p_j+d_i,\ p_i+d_j\}$ and $B=\max_{i,j}B_{ij}$, with $B=0$ when there are no rules. If $n\ge B_{ij}$ and both rules are enabled, each remains enabled after the other, so the peak commutes at $n-d_i-d_j$. Only the finitely many source lengths below $B$ remain. Their descendant graphs are finite because every step lowers length, so joinability can be checked by finite search. Newman's lemma (Theorem~\ref{thm:newman}) completes the decision procedure.
\end{proof}

\subsection{Bounded Rule Lengths}

Let $\mathcal C_{5,4,12}$ contain exactly the natural-number codes that decode as a finite declared alphabet, a finite nonerasing rule set, and a positive additive certificate such that every left side has length at most $5$, every right side has length at most $4$, every symbol weight lies in $\{1,2,3,7,12\}$, every rule strictly decreases weight, and the induced cycle relation is confluent. Malformed codes and codes failing any displayed condition lie outside $\mathcal C_{5,4,12}$. Thus certificate validity and all numerical bounds are part of the decidable input syntax, not an external promise.

\begin{proposition}[Bounded rule lengths and weights]\label{prop:bounded-shape}
The set $\mathcal C_{5,4,12}$ is $\Pi^0_1$-complete. In the instances produced by the reduction, rule lengths and symbol weights are uniformly bounded, while the alphabet and number of rules may grow.
\end{proposition}

\begin{proof}
Use the fixed normal form of Theorem~\ref{thm:verifier-source}: $k=2$ and $|\lambda|=|\rho|=2$. Inspect the wrapper rules. Launch has shape $2\to4$; decoder completion has shape $2\to4$; a machine or full-window rejection left side has length $2k+1=5$, and a machine right side has at most $2k=4$ symbols. Every remaining family (decoder moves, malformed launch and decoder, boundary, collision, sweep, acceptance, absorption) has both sides of length at most four. The certificate in Lemma~\ref{lem:wrapper-weight} takes values $1$ on $\#,A_0,A_1,N,E$, $2$ on data and copied states, $3$ on raw letters, $7=2|\lambda|+3$ on the decoders, and $12=7+2|\rho|+1$ on $X$.

The weighted reduction therefore gives $\Pi^0_1$-hardness. For the upper bound, invalid syntax or certificates are decidable. On a valid presentation, a length-$n$ source has total weight at most $12n$, descendant graphs of enumerated local peaks are finite and computable by Proposition~\ref{prop:certificate}, and a nonconfluent system has a nonjoinable local peak by Theorem~\ref{thm:newman}, so certified nonconfluence is computably enumerable. Hence certified confluence is in $\Pi^0_1$.
\end{proof}

\section{Conclusion and Future Work}\label{sec:conclusion}

We have shown that confluence of finite nonerasing strictly length-reducing cycle-rewriting systems is $\Pi^0_1$-complete over every fixed alphabet with at least two letters and decidable over a single letter. The alphabet boundary is settled by Corollary~\ref{cor:all-alphabets} and Theorem~\ref{thm:unary}. Rule shape and rule count are not. Our binary construction fixes the alphabet but uses rule lengths that grow with the encoded instance, while the bounded-shape variant of Proposition~\ref{prop:bounded-shape} bounds rule lengths and weights but lets the alphabet and rule count grow. It remains open whether confluence is undecidable for strictly length-reducing cycle systems with a fixed bound on both sides of every rule. Such a family would necessarily use a growing alphabet, since a fixed alphabet with bounded rule lengths admits only finitely many systems. Monadic restrictions, such as right-hand sides of length at most one, are another natural case. Every right-hand side produced by the binary encoding has length at least two (Lemma~\ref{lem:rhs-two}), so our construction does not cover them.

A second question concerns the number of rules. It is open whether undecidability can be obtained with a fixed number of binary length-reducing rules. For comparison, Matiyasevich and S\'enizergues proved several decision problems on three-rule semi-Thue systems undecidable \cite{MatiyasevichSenizergues2005}.

Finally, the negative result does not rule out useful positive criteria. A complementary direction is to identify syntactic or semantic restrictions under which finite critical-pair analysis becomes complete for cycle rewriting. A further direction is to make the passage from cycles to graphs precise. An effective embedding of cycle rewriting into a standard graph-transformation formalism that preserves and reflects confluence, while also preserving finiteness and termination, would pin down its formal relationship with the known graph-rewriting undecidability results \cite{Plump1993,Plump2005}.

\section*{Acknowledgements}

Claude Fable 5 (Anthropic) and GPT-5.6 Sol (OpenAI) played a major part in this work throughout. They contributed to the original discovery of the undecidability result and to its strengthening, from undecidability with a large alphabet and termination to $\Pi^0_1$-completeness with a fixed binary alphabet and strict length reduction, which I subsequently developed and verified in full. They were further used to check the mathematical development for errors, through adversarial review and through mechanical validation against an independent reference implementation, to formally verify the result in Isabelle, and to assist with proofreading and restructuring the manuscript. Proofs were checked by me, and I take full responsibility for the correctness and content of the paper.

\appendix
\renewcommand{\thesection}{\Alph{section}}
\renewcommand{\theHsection}{appendix.\Alph{section}}

\section{The Exact Generator Catalogue}\label{app:catalogue}

This appendix records the exact generator data for the command graphs of Section~\ref{sec:command-table}: the command schemata of each macro, one clause per accepted case, and the per-state sets $\mathrm{lead}(q)$ and $\mathrm{trail}(q)$. The generator of that section expands over exactly these clauses and sets to the transition table of $V_P$, one instance per admissible key.

\paragraph{Boot commands.}
Boot states are the scan states $q_{\mathrm{scan}}^{\varphi}$, whose superscript $\varphi\in\{\mathrm B\}\cup\{(i,t):i\in\{0,1\},\,0\le t\le v_i\}\cup\{\mathrm M,\mathrm F\}$ records the field index, the budget field, block $i$ with $t$ cells of its tag prefix $1^{v_i}$ matched, the interblock gap, or the final gap, and the single return state $q_{\mathrm{ret}}$; the start state $q_0$ is a restricted entry copy of $q_{\mathrm{scan}}^{\mathrm B}$, and printed boot rules suppress the superscript.

The context pairs are $\mathtt{LL}$ at entry, $\mathtt L\,f$ near $\lambda$, and otherwise $f\,c$, $c\,f$, or $f\,f$ with $c$ the last logical cell crossed, mirrored at $q_{\mathrm{ret}}$.

The clauses are: at $q_0$, only on the key $\mathtt{LL}\,q_0\,ff$: \textup{(T1)}, successor $q_{\mathrm{scan}}^{\mathrm B}$, expanding to the single instance $\mathtt{LL}\,q_0\,ff\to\mathtt{LL}\,f\,q_{\mathrm{scan}}^{\mathrm B}$, so a completing boot pass has traversed a nonempty leading reservoir; at every $q_{\mathrm{scan}}^{\varphi}$, on the key $x\,y\,q\,f\,f$: \textup{(T1)}, successor $q_{\mathrm{scan}}^{\varphi}$; at $q_{\mathrm{scan}}^{\mathrm B}$, on $x\,y\,q\,1\,f$: \textup{(T1)}, successor $\hat q_{\mathrm{scan}}^{\mathrm B}$, and on $x\,y\,q\,\mathtt L\,f$: \textup{(T1)}, successor $\hat q_{\mathrm{scan}}^{0,0}$; at $q_{\mathrm{scan}}^{i,t}$ with $t<v_i$, on $x\,y\,q\,1\,f$: \textup{(T1)}, successor $\hat q_{\mathrm{scan}}^{i,t+1}$; at $q_{\mathrm{scan}}^{i,v_i}$, on $x\,y\,q\,0\,f$: \textup{(T1)}, successor $\hat q_{\mathrm{scan}}^{i,v_i}$, and on $x\,y\,q\,\mathtt R\,f$: \textup{(T1)}, successor $\hat q_{\mathrm{scan}}^{\mathrm M}$ for $i=0$ and $\hat q_{\mathrm{scan}}^{\mathrm F}$ for $i=1$; at $q_{\mathrm{scan}}^{\mathrm M}$, on $x\,y\,q\,\mathtt L\,f$: \textup{(T1)}, successor $\hat q_{\mathrm{scan}}^{1,0}$; at $q_{\mathrm{scan}}^{\mathrm F}$, on $f\,f\,q\,\mathtt{RR}$: the turn, mirrored \textup{(T3)}, successor $q_{\mathrm{ret}}$; at $q_{\mathrm{ret}}$, on $f\,f\,q\,y\,x$: mirrored \textup{(T1)}, successor $q_{\mathrm{ret}}$, and on $f\,c\,q\,y\,x$ with $c\in\{1,0,\mathtt L,\mathtt R\}$: mirrored \textup{(T1)}, successor $\hat q_{\mathrm{ret}}$; and at $q_{\mathrm{ret}}$, on $\mathtt{LL}\,q\,ff$: parking \textup{(T3)}, successor $s_{p_{\mathrm{init}}}$, with the final state in its place when $p_{\mathrm{init}}=h$. Expansion yields, in particular, the turn instance $f\,f\,q_{\mathrm{scan}}\,\mathtt{RR}\to f\,q_{\mathrm{ret}}\,\mathtt{RR}$ and the parking instance $\mathtt{LL}\,q_{\mathrm{ret}}\,ff\to\mathtt{LL}\,s_{p_{\mathrm{init}}}\,f$, and every boot command consumes only fuel.

Until the parking exit the state lies strictly between $\lambda$ and the fence it turns at, moving right in scan states and left in $q_{\mathrm{ret}}$, and the data from $\lambda$ to a scan state, with fuel deleted, is the prefix of $1^{b}\,\mathtt L\,1^{v_0}0^{R_0-v_0}\,\mathtt R\,\mathtt L\,1^{v_1}0^{R_1-v_1}\,\mathtt R$ ending at the remembered position $\varphi$. The pass continues only on the keys listed at the current state: a token or $\mathtt L$ at $\mathrm B$, a tag $1$ at $(i,t)$ with $t<v_i$, a reserve $0$ or $\mathtt R$ at $(i,v_i)$, $\mathtt L$ at $\mathrm M$, the fence at $\mathrm F$, and at $q_{\mathrm{ret}}$ a mirrored crossing of surviving cells or, directly right of $\lambda$ against a rear $ff$, the parking key. A wrong marker, a $1$ in a reserve zone, a tag prefix of the wrong length, or a reservoir of the wrong parity therefore leaves the boot pass stuck. These are invariants of the pass from $q_0$, not of an arbitrarily placed boot state: an arbitrary candidate core can carry an interior $\mathtt{LL}$, at which a placed $q_{\mathrm{ret}}$ with rear $ff$ would park the hub, or the final state, inside the core, a window the pass never presents.

\paragraph{Pay and seek commands.}
Pay uses no constructors beyond the hub $s_p$ and the pay-active $t_p$; seek states are $c_{p,i}^{\phi,\ell}$, with $i$ the target block of the instruction at $p$, field class $\phi\in\{\mathrm B,\mathrm C,\mathrm G\}$ for the budget field, the interior of block $0$, and the interblock gap, and $\ell\in\{0,1\}$ the number of $\mathtt L$s passed, which the field class already determines with two counters and which is kept only to match the memory column of Table~\ref{tab:compiler}; the seek entry $c_p$ is the entry copy of $c_{p,i}^{\mathrm B,0}$.

The pay node carries four commands: at $s_p$, on $\mathtt{LL}\,s_p\,ff$: \textup{(T3)}, successor $t_p$; at $t_p$, on $\mathtt{LL}\,t_p\,ff$: \textup{(T3)}, successor $t_p$; at $s_p$, on $\mathtt{LL}\,s_p\,f1$: \textup{(T3)}, successor $c_p$; at $t_p$, on $\mathtt{LL}\,t_p\,f1$: \textup{(T3)}, successor $c_p$; they expand to the instances $\mathtt{LL}\,s_p\,ff\to\mathtt{LL}\,t_p\,f$, $\mathtt{LL}\,t_p\,ff\to\mathtt{LL}\,t_p\,f$, $\mathtt{LL}\,s_p\,f1\to\mathtt{LL}\,c_p\,f$, and $\mathtt{LL}\,t_p\,f1\to\mathtt{LL}\,c_p\,f$, shrinking the leading reservoir through the pay-active state and, at the token deletion, handing over to seek.

The seek clauses are: at each interior $c_{p,i}^{\phi,\ell}$, on $x\,y\,q\,f\,f$: \textup{(T1)}, successor $c_{p,i}^{\phi,\ell}$, and the same cases from the entry copy $c_p$, each with the successor its interior clause names, beginning the direct traversal of the merged reservoir; at $c_{p,i}^{\mathrm B,0}$, on $x\,y\,q\,1\,f$: \textup{(T1)}, successor $\hat c_{p,i}^{\mathrm B,0}$, crossing one remaining budget token; at $c_{p,1}^{\mathrm B,0}$, on $x\,y\,q\,\mathtt L\,f$: \textup{(T1)}, successor $\hat c_{p,1}^{\mathrm C,1}$; at $c_{p,1}^{\mathrm C,1}$, on $x\,y\,q\,1\,f$ or $x\,y\,q\,0\,f$: \textup{(T1)}, successor $\hat c_{p,1}^{\mathrm C,1}$, and on $x\,y\,q\,\mathtt R\,f$: \textup{(T1)}, successor $\hat c_{p,1}^{\mathrm G,1}$; and the handoff, at $c_{p,0}^{\mathrm B,0}$ and at $c_{p,1}^{\mathrm G,1}$, on $x\,y\,q\,\mathtt L\,f$: \textup{(T1)}, successor the flagged update entry of the instruction at $p$, inside block $0$, respectively block $1$.

During payment the state faces right at the hub position directly after $\lambda$ and never moves; its only windows are the four keys above, each a stationary deletion keeping one fuel cell visible after $\mathtt{LL}$, and the phase stops at the token deletion, which leaves $c_p$ adjacent to the merged reservoir.

During seek the state faces right, moves by \textup{(T1)} only, and lies strictly left of the fence inside the field its class names, $\mathrm B$ before the first $\mathtt L$, $\mathrm C$ inside block $0$, $\mathrm G$ between block $0$'s $\mathtt R$ and block $1$'s $\mathtt L$, having crossed exactly $\ell$ markers; an unflagged state admits fuel or a cell its class accepts, a flagged state only the fuel continuation, which enforces the ledger's parities. Payment sticks at the exhausted-budget windows $\mathtt{LL}\,s_p\,f\mathtt L$ and $\mathtt{LL}\,t_p\,f\mathtt L$, and seek on any cell its field class forbids and at any wrong-parity window $x\,y\,q\,f\,c$ with $c\ne f$.

\paragraph{Update and return commands.}
The update and return constructors and the abbreviation convention for their printed instances are declared in Section~\ref{sec:command-table}.

\emph{Increment $i$}, right-moving, its update entry the flagged entry copy of $a_{p'}$, entered by seek's crossing of $\mathtt L$: at $a_{p'}$, on $x\,y\,q\,f\,f$: \textup{(T1)}, successor $a_{p'}$; on $x\,y\,q\,1\,f$: \textup{(T1)}, successor $\hat a_{p'}$; on $x\,y\,q\,0\,f$: the \textup{(T2)} instance $x\,y\,q\,0\,f\to x\,y\,q'\,1$, successor the return entry $r^{\phi}_{p'}$ left of the new $1$, with $\phi$ the used class of block $i$; the window $x\,y\,q\,\mathtt R\,f$ has no clause, so a full block sticks.

\emph{Conditional $i$}, right-moving, its update entry the flagged entry copy of $b^{0}$: at $b^{\beta}$, on $x\,y\,q\,f\,f$: \textup{(T1)}, successor $b^{\beta}$; on $x\,y\,q\,1\,f$: \textup{(T1)}, successor $\hat b^{1}$, so crossing a used $1$ sets the bit; at $b^{0}$, on $x\,y\,q\,c\,f$ with $c\in\{0,\mathtt R\}$: \textup{(T1)}, successor the return entry $r^{\phi}_{p_0}$ right of the dispatch cell, with $\phi$ the class of the dispatch cell; at $b^{1}$, on the same keys: \textup{(T1)}, successor the turnaround state $u_{p_1}$, both reversals entering unflagged.

\emph{Turnaround}, left-moving: at $u_{p_1}$, on $f\,c\,q\,y\,x$ with $c\in\{0,\mathtt R\}$ the dispatch cell: mirrored \textup{(T1)}, successor $\hat u_{p_1}$; on $f\,f\,q\,y\,x$: mirrored \textup{(T1)}, successor $u_{p_1}$; on $f\,1\,q\,y\,x$: the mirrored \textup{(T2)} instance $f\,1\,q\,y\,x\to 0\,q'\,y\,x$, rewriting the rightmost used $1$, successor the return entry $r^{\phi}_{p_1}$ right of the new $0$, with $\phi$ the reserve class of block $i$, in which the rewritten cell now stands.

\emph{Return}, left-moving, each return entry the unflagged entry copy of $r^{\phi}_p$ at its bound class: at $r^{\phi}_p$, on $f\,f\,q\,y\,x$: mirrored \textup{(T1)}, successor $r^{\phi}_p$; on $f\,c\,q\,y\,x$ with $c$ a cell the mirrored field grammar admits after $\phi$: mirrored \textup{(T1)}, successor $\hat r^{\phi'}_p$ with $\phi'$ updated by $c$; and in the leftmost class, on $\mathtt{LL}\,q\,ff$: the parking \textup{(T3)} instance $\mathtt{LL}\,r^{\phi}_p\,ff\to\mathtt{LL}\,s_p\,f$, with the designated final state in place of $s_p$ when $p=h$.

The increment faces right inside block $i$ and remembers $p'$; its \textup{(T2)} key fires unflagged and places the return entry immediately left of the rewritten cell, that cell followed by surviving fuel exactly when the consumption leaves the reservoir beyond nonempty; the exhausted case sticks at the entry, as recorded with the return entries below. The conditional faces right between the block's $\mathtt L$ and its dispatch cell; $\beta=1$ exactly when a used $1$ was crossed, hence, the used zone being contiguous, exactly when $m_i>0$; its dispatch keys fire only unflagged, and the dispatch consumes the fuel beyond the dispatch cell, as every \textup{(T1)} crossing does.

The turnaround re-crosses the dispatch cell by consuming the surviving cell left of it, then meets only surviving fuel, so its first non-fuel window sits at the rightmost used $1$, whose left-hand fuel survives because the rightward pass crossed that $1$ from an unflagged state. The return performs no rewrite, so within it the data changes only by fuel consumption, and its parking key demands two visible fuel cells, so parking never consumes that reservoir's final cell. Any other window, a misplaced marker, $\mathtt R$ under an increment at full capacity, a flagged state facing a logical cell, or fuel of the wrong parity, leaves the phase stuck.

\emph{Boot.} The entry has $\mathrm{trail}(q_0)=\{\mathtt{LL}\}$ and $\mathrm{lead}(q_0)=\{ff\}$, the single key of its sole command. The budget scan has $\mathrm{trail}(q_{\mathrm{scan}}^{\mathrm B})=\{\mathtt Lf,\,ff,\,1f\}$, the near-$\lambda$ pair left by $q_0$, the fuel pair, and the flag-cleared token, with $\mathrm{trail}(\hat q_{\mathrm{scan}}^{\mathrm B})=\{f1\}$. Inside block $i$ with $t$ tag cells matched, let $C$ hold the last cell the class crossed, $\mathtt L$ for $t=0$ and $1$ for $t\ge1$, together with $0$ when $t=v_i$, both applying at $v_i=0$; then $\mathrm{trail}(q_{\mathrm{scan}}^{i,t})=\{cf:c\in C\}\cup\{ff\}$ and $\mathrm{trail}(\hat q_{\mathrm{scan}}^{i,t})=\{fc:c\in C\}$.

The gap scans have $\mathrm{trail}(q_{\mathrm{scan}}^{\mathrm M})=\mathrm{trail}(q_{\mathrm{scan}}^{\mathrm F})=\{\mathtt Rf,\,ff\}$ with hatted $\{f\mathtt R\}$; only the rear $ff$ carries the turn, so a final run pairing to a single survivor sticks at $\mathtt Rf\,q\,\mathtt{RR}$. The return has $\mathrm{trail}(q_{\mathrm{ret}})=\{\mathtt{RR},\,f\mathtt R,\,ff,\,f\mathtt L,\,f0,\,f1\}$, the fence pairs at the turn end, then survivors and flag-cleared re-crossed cells, and $\mathrm{trail}(\hat q_{\mathrm{ret}})=\{\mathtt Rf,\,\mathtt Lf,\,0f,\,1f\}$; the turn initializes the rear $\mathtt{RR}$, against which the fuel-pair key opens the return when two fuel cells remain to the left, later steps turning the rear to $f\mathtt R$ and then $ff$, while a single remaining cell leaves no key; the turn leaves fuel adjacent, so no logical crossing fires there and the flagged copy carries no fence pair; and parking fires only on the rear $ff$, so a $\lambda$-run at its last survivor sticks at $\mathtt{LL}\,q_{\mathrm{ret}}\,fc$.

\emph{Pay and seek.} The pay states have $\mathrm{trail}(s_p)=\mathrm{trail}(t_p)=\{\mathtt{LL}\}$ and $\mathrm{lead}(s_p)=\mathrm{lead}(t_p)=\{ff,\,f1\}$, whose products are exactly the four pay deletions. The seek entry shares $\mathrm{trail}(c_p)=\{\mathtt{LL}\}$ and, as the entry copy of $c_{p,i}^{\mathrm B,0}$, has $\mathrm{lead}(c_p)=\{ff,\,1f,\,\mathtt Lf\}$, each key with the successor its interior clause names. Since the merged reservoir is nonempty, only $c_p$'s fuel key ever matches a grammar-shaped window, so no rear $\mathtt L1$ reaches $\hat c_{p,i}^{\mathrm B,0}$ and no rear $\mathtt{LL}$ reaches $\hat c_{p,1}^{\mathrm C,1}$. The seek interiors repeat the boot pattern: $\mathrm{trail}(c_{p,i}^{\mathrm B,0})=\{\mathtt Lf,\,ff,\,1f\}$ with hatted $\{f1\}$, $\mathrm{trail}(c_{p,1}^{\mathrm C,1})=\{\mathtt Lf,\,ff,\,1f,\,0f\}$ with hatted $\{f\mathtt L,\,f1,\,f0\}$, and $\mathrm{trail}(c_{p,1}^{\mathrm G,1})=\{\mathtt Rf,\,ff\}$ with hatted $\{f\mathtt R\}$.

{\sloppy
\emph{Update.} Right-moving inside block $i$: $\mathrm{trail}(a_{p'})=\{\mathtt Lf,\,ff,\,1f\}$ with $\mathrm{lead}(a_{p'})=\{ff,\,1f,\,0f\}$ and no $\mathtt Rf$ clause, hatted trail $\{f1\}$; $\mathrm{trail}(b^{0})=\{\mathtt Lf,\,ff\}$ and $\mathrm{trail}(b^{1})=\{ff,\,1f\}$ with $\mathrm{lead}(b^{\beta})=\{ff,\,1f,\,0f,\,\mathtt Rf\}$ and hatted $\mathrm{trail}(\hat b^{1})=\{f1\}$; the flagged update entries have trail $\{f\mathtt L\}$, the freshly crossed opener.\par}

{\sloppy
\emph{Turnaround and return classes.} Left-moving, rear pairs to the right: $\mathrm{trail}(u_{p_1})=\{ff,\,f0,\,f\mathtt R,\,f\mathtt L\}$, at entry a survivor and the sequel of the reservoir beyond the dispatch cell, thereafter a survivor and the dispatch cell or fuel, with $\mathrm{lead}(u_{p_1})=\{ff,\,f0,\,f\mathtt R,\,f1\}$ and $\mathrm{trail}(\hat u_{p_1})=\{0f,\,\mathtt Rf\}$. The return classes, leftward, are $\mathrm F,\mathrm E_1,\mathrm U_1,\mathrm G,\mathrm E_0,\mathrm U_0,\mathrm B$: the final gap before the fence, the reserve and used zones of block $1$, the interblock gap, those of block $0$, and the budget field. The walk admits, at $\mathrm F$, $\mathtt R\mapsto\mathrm E_1$; at $\mathrm E_i$, $0\mapsto\mathrm E_i$, $1\mapsto\mathrm U_i$, and $\mathtt L\mapsto\mathrm G$ for $i=1$, $\mathtt L\mapsto\mathrm B$ for $i=0$; at $\mathrm U_i$, $1\mapsto\mathrm U_i$ and $\mathtt L$ as at $\mathrm E_i$; at $\mathrm G$, $\mathtt R\mapsto\mathrm E_0$; at $\mathrm B$, $1\mapsto\mathrm B$, together with the parking key on the scanned $\mathtt{LL}$ against the rear $ff$. Interior states carry these sets exactly where the closure of Section~\ref{sec:command-table} generates them. An unflagged $r^{\phi}_p$, at each class bound by one of $p$'s creation sites or reached from one along the walk, carries $\mathrm{lead}(r^{\phi}_p)=\{ff\}\cup\{fc:\phi\ \text{admits}\ c\}$, together with the parking pair $\mathtt{LL}$ at $\phi=\mathrm B$, generated only against the rear $ff$, and $\mathrm{trail}(r^{\phi}_p)=\{ff\}\cup\{fc:c\ \text{enters}\ \phi\}$. A flagged $\hat r^{\phi}_p$, exactly where a generated unflagged return state crosses into $\phi$, carries $\mathrm{trail}(\hat r^{\phi}_p)=\{cf:c\ \text{enters}\ \phi\}$ with $\mathrm{lead}(\hat r^{\phi}_p)=\{ff\}$; every other interior tuple has an empty trail set and receives no keys, as $\hat r^{\mathrm F}_p$ always and $\hat r^{\mathrm G}_p$ when $\mathrm G$ is the rightmost class bound for $p$.\par}

\emph{Return entries.} The increment binds $\mathrm U_i$; its entry, left of the new $1$, has trail $\{1f\}$ exactly, the rewritten cell followed by surviving fuel, so the pairs $10$ and $1\mathtt R$, left behind when the \textup{(T2)} consumed that reservoir's last cell, are inadmissible and the execution sticks at the entry; certification depends on this exclusion.

The zero branch binds the class of the dispatch cell, $\mathrm E_i$ for a reserve $0$, $\mathrm G$ for block $0$'s $\mathtt R$, $\mathrm F$ for block $1$'s $\mathtt R$; its entry demands a leading survivor, trail $\{ff,\,f0,\,f\mathtt R\}$ at $\mathrm E_i$, $\{ff,\,f\mathtt L\}$ at $\mathrm G$, $\{ff,\,f\mathtt R\}$ at $\mathrm F$, so a dispatch that emptied the reservoir beyond sticks at the entry.

The turnaround binds $\mathrm E_i$, the reserve class in which its rewritten cell now stands, so the walk supplies the re-crossing of the new $0$ natively; its entry has trail $\{ff,\,f0,\,f\mathtt R\}$, and an exhausted or even retrace never reaches it, sticking at $\hat u_{p_1}$ when at most two survivors precede the dispatch cell and otherwise at the unflagged $u_{p_1}$ on the wrong-parity window $1f$; a \textup{(T2)} that consumed the sole survivor left of the rewritten $1$ leaves the entry scanning $10$ or $\mathtt L0$, outside its lead, so it sticks at once. Creation sites binding the same class and target share one entry copy. Expanding the generator over these sets, one instance per admissible key, yields exactly the printed clauses and nothing else.

\section{Detailed Verifier Calculations}\label{app:calculations}

We trace one conditional's decrement branch through payment, seek, dispatch, turnaround, update, and return on a compact instance exercising all six phases: at control $p$ the instruction $p:\ \mathrm{if}\ C_0=0\ \mathrm{goto}\ p_0$ $\mathrm{else}\ C_0:=C_0-1;\ \mathrm{goto}\ p_1$, with $B=1$, $R_0=2$, $R_1=1$, $m_0=1$, $m_1=0$, so the pass sets the seen-$1$ bit and dispatches at the first reserve $0$. Index this witness's nine reservoirs $\mu_1,\dots,\mu_9$ from $\lambda$ rightward. The chronological ledger trace is as follows.
\begin{center}
\small
\setlength{\tabcolsep}{4pt}
\begin{tabular}{@{}llrr@{}}
\hline
Event & Reservoir & Before & After \\
\hline
payment merge, leading $\mu_1$ & $(\mu_1,\mu_2)$ & $(1,9)$ & $10$ \\
direct traversal by seek & merged & $10$ & $5$ \\
seam, entered crossing $\mathtt L$ & $\mu_3$ & $9$ & $4$ \\
seam, entered crossing the used $1$ & $\mu_4$ & $7$ & $3$ \\
single consumption at the dispatch crossing & $\mu_5$ & $2$ & $1$ \\
seam, turnaround re-cross & $\mu_4$ survivors & $3$ & $1$ \\
single consumption, mirrored \textup{(T2)} & $\mu_3$ survivors & $4$ & $3$ \\
return seam over the new $0$ & $\mu_3$ survivors & $3$ & $1$ \\
return seam over $\mathtt L$ & merged survivors & $5$ & $2$ \\
parking at $\lambda$ & merged survivors & $2$ & $1$ \\
\hline
\end{tabular}
\end{center}
Unit terminal reservoirs and the canonical payment preimage with leading size $1$ give the assignment $(1,9,9,7,2,1,1,1,1)$: invert in reverse order per Lemma~\ref{lem:fuel-alloc}, by $m\mapsto m+1$ three times, $m\mapsto2m+1$ five times, $m\mapsto2m$ once, and the merge split $10=1+9$; each seam preimage is odd, the direct preimage is even, parking finds its two fuels, and the unvisited $\mu_6,\dots,\mu_9$ stay at $1$. The inverted effects are exact and payment completes from every positive leading size, so the positive assignments whose pass ends with every reservoir a single cell are exactly $(a,9,9,7,2,1,1,1,1)$ with $a\ge1$, and this one, $a=1$, is their componentwise minimum. The certified boundary at the hub is therefore
\[
  \mathtt{LL}\;s_p\;
  f\,1\,f^{9}\,
  \mathtt L\,f^{9}\,1\,f^{7}\,0\,f^{2}\,\mathtt R\,
  f\,\mathtt L\,f\,0\,f\,\mathtt R\,f\;\mathtt{RR}.
\]
In stages, suppressing state superscripts: (1) the boundary above; (2) with $\mu_1=1$ payment exits directly by $\mathtt{LL}\,s_p\,f1\to\mathtt{LL}\,c_p\,f$, never entering $t_p$, and the token deletion leaves the merged $f^{10}$; (3) seek pairs it to five survivors and crosses $\mathtt L$; (4) the conditional pairs $\mu_3$ to four survivors, crosses the used $1$ setting its bit, and pairs $\mu_4$ to three; (5) the dispatch key fires as $f\,f\,b^{1}\,0\,f\to f\,f\,0\,u_{p_1}$; (6) the turnaround re-crosses the $0$, retraces to one survivor, and updates by the mirrored \textup{(T2)} instance $f\,1\,u_{p_1}\,f\,0\to0\,r_{p_1}\,f\,0$, branching towards $p_1$; (7) the return re-crosses the new $0$ and $\mathtt L$, pairing $f^{3}$ and $f^{5}$ to one and two survivors; (8) parking fires as $\mathtt{LL}\,r_{p_1}\,ff\to\mathtt{LL}\,s_{p_1}\,f$, leaving $\mathtt{LL}\,s_{p_1}\,f\,\mathtt L\,f\,0\,f\,0\,f\,\mathtt R\,f\,\mathtt L\,f\,0\,f\,\mathtt R\,f\,\mathtt{RR}$, a certified boundary with every reservoir a single cell. The eight stages stand after $0$, $1$, $7$, $15$, $16$, $19$, $24$, and $25$ generated-table steps, successive stages thus $1$, $6$, $8$, $1$, $3$, $5$, and $1$ commands apart. Fuel deleted and the state omitted, the pass maps the core $1\,\mathtt L\,1\,0\,\mathtt R\,\mathtt L\,0\,\mathtt R$ to $\mathtt L\,0\,0\,\mathtt R\,\mathtt L\,0\,\mathtt R$: the budget loses its one token and $m_0=1$ becomes $0$, matching $\pi$'s step from $(p,1,0)$ to $(p_1,0,0)$.

\bibliographystyle{elsarticle-num}
\bibliography{cycle_confluence_undecidability_tcs}

\end{document}